\documentclass[a4paper,reqno, 10pt]{amsart}

\usepackage[T1]{fontenc}

\usepackage[margin=2cm]{geometry}
\usepackage[foot]{amsaddr}
\usepackage[colorlinks, linkcolor = blue, citecolor = blue, urlcolor = blue]{hyperref}

\usepackage{amsmath,amsthm,amssymb,amsfonts}
\usepackage{ytableau}
\usepackage{float}
\usepackage{adjustbox}
\usepackage{mathtools}

\usepackage{thmtools}

\usepackage{mathrsfs}

\usepackage{cleveref}

\usepackage[shortlabels]{enumitem}

\usepackage{makecell}
\usepackage{multirow}
\usepackage{xcolor}
\usepackage{hhline}
\usepackage{diagbox}

\usepackage{caption}
\usepackage{subcaption}

\usepackage[group-separator={,},group-minimum-digits=4]{siunitx}

\usepackage{epigraph} 

\usepackage{graphicx}
\usepackage{tikz}
\usetikzlibrary{quantikz2}
\usetikzlibrary{backgrounds}
\usetikzlibrary{external}
\usetikzlibrary{calc}
\usetikzlibrary{shapes}

\usepackage{pgfplots}
\pgfplotsset{compat=1.18} 

\usepackage{booktabs}

\usepackage{ytableau,varwidth}

\usepackage[backend=biber,style=alphabetic,maxnames=9,maxalphanames=5,isbn=false,backref]{biblatex}

\usepackage[backgroundcolor=pink]{todonotes}
\usepackage{marginnote}

\newtheorem*{theorem*}{Theorem}
\newtheorem{theorem}{Theorem}
\newtheorem{lemma}[theorem]{Lemma}

\crefname{lemma}{Lemma}{Lemmas}
\crefname{definition}{Definition}{Definitions}
\crefname{theorem}{Theorem}{Theorems}
\crefname{conjecture}{Conjecture}{Conjectures}
\crefname{section}{Section}{Sections}
\crefname{claim}{Claim}{Claims}
\crefname{appendix}{Appendix}{Appendices}
\crefname{figure}{Fig.}{Figs.}
\crefname{table}{Table}{Tables}
\crefname{proposition}{Proposition}{Propositions}
\crefname{corollary}{Corollary}{Corollaries}
\crefname{example}{Example}{Examples}
\crefname{remark}{Remark}{Remarks}

\providecommand\given{}
\newcommand\SetSymbol[1][]{%
    \nonscript\:#1\vert
    \allowbreak
    \nonscript\:
    \mathopen{}}
\DeclarePairedDelimiterX\Set[1]\{\}{%
    \renewcommand\given{\SetSymbol[\delimsize]}
    #1
}

\DeclarePairedDelimiter{\set}{\lbrace}{\rbrace}

\DeclarePairedDelimiter{\norm}{\lVert}{\rVert}

\DeclarePairedDelimiter{\of}{\lparen}{\rparen}
\DeclarePairedDelimiter{\sof}{\lbrack}{\rbrack}

\newcommand{\defeq}{\vcentcolon=}

\renewcommand{\leq}{\leqslant}
\renewcommand{\geq}{\geqslant}

\renewcommand{\bra}[1]{\langle{#1}\rvert}
\renewcommand{\ket}[1]{\lvert{#1}\rangle}
\renewcommand{\braket}[2]{\langle{#1}|{#2}\rangle}
\newcommand{\ketbra}[2]{\ket{#1}\bra{#2}}
\renewcommand{\proj}[1]{\ketbra{#1}{#1}}

\newcommand{\ct}{^{\dagger}}
\newcommand{\tp}{^{\mathsf{T}}}

\newcommand{\x}{\otimes}
\newcommand{\xp}[1]{^{\otimes #1}}

\newcommand{\C}{\mathbb{C}} 
\newcommand{\cont}{\mathrm{cont}} 
\newcommand{\wcont}{\mathrm{wcont}} 
\newcommand{\Brat}{\mathscr{B}} 
\newcommand{\A}{\mathcal{A}} 

\newcommand{\0}{\varnothing}

\let\S\relax
\DeclareMathOperator{\S}{S} 
\DeclareMathOperator{\Tr}{Tr} 

\newcommand{\AC}{\mathrm{AC}} 
\newcommand{\RC}{\mathrm{RC}} 
\newcommand{\Irr}[1]{\widehat#1} 
\newcommand{\Paths}{\mathrm{Paths}} 

\newcommand{\U}[1]{\mathrm{U}_{#1}} 

\newcommand{\pt}{\mathbin{\vdash}} 

\newcommand{\w}{0.5cm}

\newcommand{\bx}[3]{
  \draw[fill = white] #3 (#1*\w-\w/2,-#2*\w-\w/2) rectangle (#1*\w+\w/2,-#2*\w+\w/2);
}

\newcommand{\yd}[2][0.4]{%
  \begin{tikzpicture}[scale = #1, baseline={([yshift=-0.6ex]current bounding box.center)}]
    \foreach \li [count = \y] in {#2} {
      \foreach \x in {1,...,\li} {
        \bx{\x}{\y}{}
      }
    }
  \end{tikzpicture}
}

\newcommand\restr[2]{{
  \left.\kern-\nulldelimiterspace 
  #1 
  \vphantom{\big|} 
  \right|_{#2} 
  }}

\newcommand{\ent}{\mathrm{ent}}

\title{Entanglement recycling in two-step port-based teleportation}

\begin{document}
\marginparwidth=1.5cm
\author{Piotr Kopszak\textsuperscript{1}}
\email{piotr.kopszak2@uwr.edu.pl}
\address{\textsuperscript{1} Institute for Theoretical Physics, University of Wroc{\l}aw, Poland}

\author{Dmitry Grinko\textsuperscript{2,3,4}\hspace{-.25em}}
\email{d.grinko@uva.nl}
\address{\textsuperscript{2} QuSoft, Amsterdam, The Netherlands}
\address{\textsuperscript{3} Institute for Logic, Language, and Computation, University of Amsterdam, The Netherlands}
\address{\textsuperscript{4} Korteweg-de Vries Institute for Mathematics, University of Amsterdam, The Netherlands}

\author{Adam Burchardt\textsuperscript{2,5}\hspace{-.25em}}
\email{adam.burchardt@cwi.nl}
\address{\textsuperscript{5} Centrum Wiskunde \& Informatica (CWI), Amsterdam, The Netherlands}

\author{Maris Ozols\textsuperscript{2,3,4}\hspace{-.25em}}
\email{marozols@gmail.com}

\author{Micha{\l} Studzi{\'n}ski\textsuperscript{6}}
\address{\textsuperscript{6} International Centre for Theory of Quantum Technologies (ICTQT), University of Gda{\'n}sk, Poland}
\email{michal.studzinski@ug.edu.pl}

\author{Marek Mozrzymas \textsuperscript{1}}
\address{}

\begin{abstract}
	A protocol involving the repetitive (twofold, to be precise) application of PBT protocol to the same resource is studied. The quantities characterizing the  resulting protocol, so-called \textit{two-step PBT}, namely \textit{enatnglement fidelity} and \textit{success probability} are provided for two scenarios, relying on application of pretty-good measurement, i.e. deterministic and probabilistic PBT with non-EPR resource. This results show that two-step PBT is an accurate protocol, provided the resource is sufficiently large. In particular, the deterministic two-step PBT obtains fidelity that is remarkably close to the optimal MPBT fidelity for teleportation of two quantum states. Additionally, the \textit{recycling fidelity}, i.e. the quantity characterizing the degradation of the resource state is calculated for repetitive application of probabilistic protocol, for both EPR and optimized resource, showing that entanglement recycling with two-step PBT is possible in the former case as well.

\end{abstract}

\maketitle
\tableofcontents

\section{Introduction}
Quantum teleportation is one of the most remarkable applications of quantum entanglement, enabling the transmission of an unknown quantum state between two parties by employing only shared entanglement and classical communication~\cite{bennett_teleporting_1993}. Despite its wide range of applications, quantum teleportation has several limitations, such as the unitary correction step on receiver's side at the end of the protocol. This raises a natural question of whether this step can be avoided.

Port-based teleportation (PBT) introduced in 2008 by Hiroshima and Ishizaka \cite{ishizaka_asymptotic_2008,ishizaka_quantum_2009} eliminates this problem.  In this variant of quantum teleportation, the sender and receiver share a resource state composed of $N$ maximally entangled EPR pairs called \emph{ports}. The sender performs a joint quantum measurement on the quantum message register and their half of the ports, and then uses classical communication to send the measurement outcome to the receiver who only needs to select one of their ports to retrieve the original quantum message. Since the overall procedure commutes with an arbitrary unitary on the receiver's port, it can be used to simulate a universal programmable quantum processor \cite{yang2020optimal,ishizaka_asymptotic_2008,ishizaka_quantum_2009}. However, due to the no-programming theorem \cite{Nielsen1997}, a PBT protocol with a finite resource state cannot be perfect and its efficiency depends on the resource state size $N$ and the port dimension $d$.

Thus, two types of protocols emerge: a \emph{deterministic inexact} one (dPBT), where the state is always transmitted with non-unit fidelity, and \emph{probabilistic exact} one (pPBT), where teleportation has unit fidelity, but the protocol has a non-zero probability of failure.  This difference between the two kinds of protocols manifests in the structure of the set of measurements available to the sender. In dPBT, the sender implements a joint $N$-outcome POVM $\{\Pi_i\}_{i=1}^N$ on the message system and their half of the resource state, getting an outcome $i \in \set{1,\dotsc,N}$ that is transmitted through a classical channel to the receiver. To recover the original quantum message, the receiver just picks the right port indicated by the classical message $i$. In pPBT, the sender has one additional measurement operator $\Pi_0$ that corresponds to failure which occurs with probability $p_{\mathrm{fail}}$.

Instead of using $N$ EPR pairs, the two parties can share an arbitrary entangled state on $N$ ports each. Such optimized resource states can lead to quadric improvement in teleportation efficiency~\cite{ishizaka_quantum_2009,Studzinski2017,StuNJP,majenz2}. Nevertheless, in the asymptotic limit $N \to \infty$, every PBT scheme achieves faithful teleportation but with different rates of convergence~\cite{majenz2}. In every version of PBT the corresponding measurements are in general different. However, except for non-optimized pPBT, the square-root measurement is known to be optimal~\cite{Leditzky2022,Grinko2024}.

Analyzing the efficiency of PBT protocols has been a challenging problem and, for a long time, a comprehensive description -- especially in higher dimensions and asymptotic regimes -- was lacking.
The main obstacle to this is the mathematical complexity of the problem and the need for advanced representation-theoretic tools. This led to significant efforts to refine and optimize different versions of the protocol and to develop appropriate mathematical frameworks~\cite{wang_higher-dimensional_2016,Studzinski2017,StuNJP,MozJPA,majenz2,Leditzky2022}, with recent advancements involving mixed Schur--Weyl duality~\cite{grinko2023gelfandtsetlin}.

Understanding the ingredients of optimal PBT protocols and their efficiency, along with the lack of the unitary correction,
has garnered significant interest. This has led to extensive studies of PBT and new advancements in quantum information theory. The PBT framework serves as a model for a universal programmable quantum processor~\cite{ishizaka_asymptotic_2008} and establishes connections with quantum cryptography and instantaneous non-local computation~\cite{beigi_konig}.
PBT protocols have played a crucial role in linking interaction complexity with entanglement in non-local computation and holography~\cite{may2022complexity}, as well as in demonstrating the relationship between quantum communication complexity advantages and Bell inequality violations~\cite{buhrman_quantum_2016}. They have also been essential in deriving fundamental limits for quantum channel discrimination through PBT stretching protocols~\cite{pirandola2019fundamental} and have contributed to various other key findings~\cite{pereira2021characterising,quintino2021quantum}. PBT serves also as a subroutine in important higher-order quantum operation methods, such as storing and retrieving quantum programs in quantum memory~\cite{PhysRevLett.122.170502}, quantum learning \cite{bisio2010optimal}, improving unitary programming methods~\cite{grosshans2024multicopyquantumstateteleportation}, the equivalence of PBT and the unitary estimation task~\cite{yoshida2024onetoonecorrespondencedeterministicportbased}, and many others~\cite{quintino2019probabilistic, quintino2019reversing, quintino2022deterministic, yoshida2023universal}. Recently, it was shown that all variants of the PBT protocol can be efficiently implemented as quantum circuits~\cite{Grinko2024,fei2023efficientquantumalgorithmportbased,PRXQuantum.5.030354}, paving a way for practical applications in quantum computing.

\begin{figure}[h]
	\centering
	\includegraphics[width=0.8\linewidth,page=1]{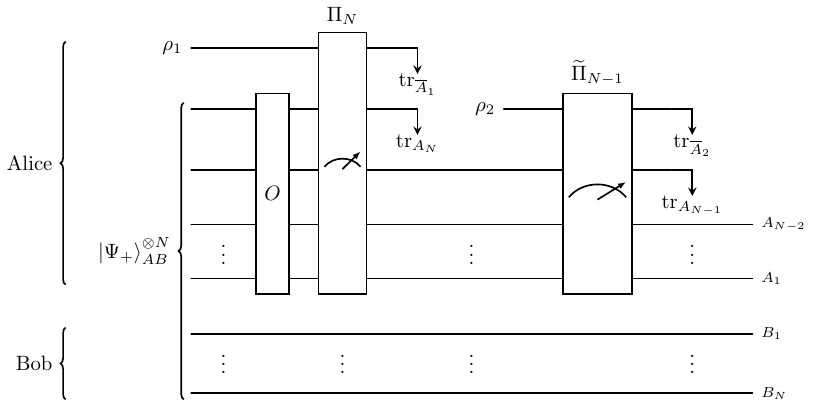}
	\caption{Two-step PBT protocol at a glance. Initially, Alice and Bob share an optimized resource state $\ket{\Psi}_{AB} = (O_A \otimes I_B) \ket{\Psi^+}_{AB}$ where registers $A = \set{A_1, \dotsc, A_N}$ and $B = \set{B_1, \dotsc, B_N}$ consist of $N$ qudits each.
		Alice's goal is to teleport qudit states $\rho_{1}=\Tr_{\bar A_2}\rho_{\bar A}$ and $\rho_{2}=\Tr_{\bar A_1}\rho_{\bar A}$ stored in her registers $\bar A_1$ and $\bar A_2$.
		First, Alice performs a POVM $\Pi$ on $\bar A_1 \cup A$.
		If the measurement outcome is $i \neq 0$(i.e. in the deterministic inexact scheme, or in the case of success in probabilistic exact), the first state ${\rho_1}$ is
		teleported to Bob's register $B_i$ (we depict the case when $i = N$). After discarding $\set{\bar A_1, A_i, B_i}$, Alice and Bob share a resource consisting of
		$N-1$ remaining registers on each side.
		In the second round, Alice performs a POVM $\widetilde{\Pi}$ on her registers $\bar A_2 \cup (\{A_1,\ldots,A_N\} \setminus A_i)$. Measuring the outcome $j \neq 0$ in this round indicates success (we depict the case when $j = N-1$). We call this protocol two-step PBT.
	}
	\label{fig:protocol}
\end{figure}

\section{Two-step PBT: description of the protocol}

Since the shared entangled resource employed in port-based teleportation is valuable, one is interested in the possibility to reuse it after one  teleportation, for another implementation of PBT protocol. In particular, a two successful executions of PBT protocol on the same resource yield a protocol enabling teleportation of two quantum states. We call this type of protocols \emph{two-step PBT}, and it is described as follows:

\begin{enumerate}
	\item Alice performs a POVM measurement $\set{\Pi_i^{A \bar{A}_1}}_{i=0}^N$ on $A \bar{A}_1$ which produces outcome $i \in \set{0,\dotsc,N}$. Let us assume success in the first round, meaning $i\neq 0$ (this holds trivially for inexact protocol, and  is reasonable for probabilistic exact one, since $p_{\mathrm{fail}}\rightarrow 0$ as $N \rightarrow \infty$).
	\item Alice sends the outcome $i$ by a classical channel to Bob who recovers the transmitted state $|\psi_1\rangle$ from his $i$-th port.
	\item Both parties apply a transposition (SWAP) between the 1st and the $i$-th port. They will not use the 1st port in subsequent round of the protocol anymore and use only the remaining $N-1$ ports; without the loss of generality one can assume that $i=N$.
	\item Alice performs the second POVM $\set{\widetilde\Pi_i^{A \bar{A}_1}}_{i=0}^{N-1}$ on $\set{A_1, \dotsc, A_{N-1}} \cup \bar{A}_2$ and communicates the outcome $i' \in \set{1,\dotsc,N-1}$ to Bob via the classical channel, implementing a PBT protocol on the remaining resource.

\end{enumerate}

This procedure can be formalized as a quantum subchannel\footnote{This map is completely positive but not necessarily trace-preserving since each step of the two-step PBT can fail with some probability. In case of failure the protocol does not produce a quantum output, which amounts to reducing the trace.}:
\begin{equation}
	\begin{split}
		\label{eq:channel_N0}
		\mathcal{N}(\rho_{\bar A}) \defeq &
		\sum_{i = 1}^N \sum_{\substack{j = 1 \\ j \neq i}}^{N}
		\Tr_{A B^c_{i,j} \bar A}
		\sof*{
			\sqrt{\widetilde{\Pi}_{j}} {\sqrt{\Pi_i}} ( \Psi_{AB} \otimes \rho_{\bar A} ) {\sqrt{\Pi_i}} \sqrt{\widetilde{\Pi}_{j}}
		}_{B_i,B_j \to \bar B_1, \bar B_2}
	\end{split}
\end{equation}
which is depicted in Figure~\ref{fig:protocol}. This protocol thus falls within the category of multiport-based teleportation protocols~\cite{stud2020A}
. We usually quantify the performance of such protocols via two figures of merit---\emph{entanglement fidelity} and \emph{average success probability}, which is usually called \textit{success probability}, although previously some particular cases of such protocols were studioed with respect to the degradation of resource state~\cite{strelchuk_generalized_2013, Studziński_2022}.

\subsection{Quantifying the performance of two-step PBT}

\textit{Entanglement fidelity}
\begin{equation}
	\label{eq:entanglement_fidelity0}
	F_{\mathrm{ent}}(\mathcal{N}) \defeq \Tr \sof*{
	\of[\big]{
	\psi^+_{\bar{B}_1R_1} \otimes \psi^+_{\bar{B}_2R_2}
	} \cdot
	\of[\big]{
	\mathcal{N}_{\bar{A}_1\bar{A}_2 \to \bar{B}_1\bar{B}_2} \otimes
	\mathcal{I}_{R_1R_2}
	}
	\sof[\big]{
	\psi^+_{\bar{A}_1R_1}\otimes \psi^+_{\bar{A}_2R_2}
	}
	},
\end{equation}
characterizes the average fidelity $f$ of the teleportation subchannel $\mathcal{N}$ as $f=\frac{dF+1}{d+1}$ where $d$ is the dimension of the teleported state.
The \textit{average probability of success} of the two-step PBT is given by
\begin{equation}
	\label{eq:p_succ_0}
	p_{\mathrm{succ}}(\mathcal{N}) \defeq \Tr
	\sof[\big]{
	\mathcal{N}_{\bar{A}_1\bar{A}_2 \to \bar{B}_1\bar{B}_2}[I/d^2]
	}.
\end{equation}

Our protocols in general are probabilistic
(i.e., $p_{\mathrm{succ}} < 1$) and inexact
(i.e., $F_{\mathrm{ent}}/p_{\mathrm{succ}} < 1$). However, there are exact probabilistic protocols
(i.e., $F_{\mathrm{ent}}/p_{\mathrm{succ}} = 1$ and $p_{\mathrm{succ}} < 1$) and inexact deterministic
(i.e., $F_{\mathrm{ent}}/p_{\mathrm{succ}} < 1$ and $p_{\mathrm{succ}} = 1$).

In this work, we concentrate on the repetitive implementation of particular POVM, the \textit{pretty good measurement} (PGM) which plays a central role in port-based teleportation schemes. In standard form, it consists of $N$ elements $\{\Pi_i\}_{i=1}^N$ and an additional element $\Pi_0$, which can either be added to the POVM yielding $\Pi=\{\Pi_i\}_{i=0}^N$ or distributed among each effect, thus introducing POMV  $\Pi^\star_i\coloneq\Pi_i+\frac{\Pi_0}{N}$. Depending on the POVM and the resource state, the resulting one-step protocol is exact probabilistic (POVM $\Pi$), inexact deterministic (POVM $\Pi^\star$). PGM is known to be optimal for one-step PBT~\cite{Leditzky2022}.

However, it is not obvious how the repetitive application of two PGMs $\Pi$ and $\widetilde\Pi$ affects entanglement fidelity and the probability of success of the two-step PBT protocol.

Note that the teleportation subchannel~\eqref{eq:channel_N0} is the same for both scenarios, since both employ PGM.
However, in the probabilistic scheme it does not preserve trace, so the entanglement fidelity corresponding to it has to be normalized by the probability of success.

\subsection{Degradation of the resource state in two-step probabilistic PBT}

The earlier approach to analyze the possibility of repetitive application of the PBT protocol \cite{strelchuk_generalized_2013, Studziński_2022} was based on the analysis of degradation of the resource state after each round of PBT teleportation, and with that approach it was possible to show that it is possible to reapply the deterministic protocol arbitrarily many times, provided the resource is large enough. However, the answer was indirect, since it does not tell how well the teleportation employing the recycled resource is performed, unlike the method described in the prevoious paragraph .

The idea of degradation is formalized via a quantity called \textit{recycling fidelity}\footnote{Note that the notion of recycling fidelity is different from entanglement fidelity.}, which is defined as the overlap between the post-teleportation resource state together with a half of the EPR pair that was teleported $\Psi_{\mathrm{out}}^{(i)}$, and an idealized state $\Psi_{\mathrm{id}}^{(i)}$, which would perfectly suit the successive application of the PBT protocol:
\begin{equation}
	\label{eq:recycling_fidelity_00}
	F_{\mathrm{rec}}(N, d) \defeq \sum_{i=0}^N p_i F(\Psi_{\mathrm{out}}^{(i)}, \Psi_{\mathrm{id}}^{(i)}),
\end{equation}
where $F(\Psi, \Phi)$ is the standard fidelity between density matrices $\Psi$ and $\Phi$. This quantity can be expressed in terms of conditional fidelities, depending on the success or failure in the preceding step
\begin{equation}
	F_{\mathrm{rec}}(N, d) = p_{\mathrm{fail}}F_{\mathrm{rec}}^{\mathrm{fail}}(N, d) + p_{\mathrm{succ}}	F_{\mathrm{rec}}^{\mathrm{succ}}(N, d)
\end{equation}
where $p_{\mathrm{fail}}=p_0, p_\mathrm{succ}=	\sum_{ i=1}^{N} p_i$ and
\begin{equation}
	F_{\mathrm{rec}}^{\mathrm{fail}}= \frac{ F(\Psi_{\mathrm{out}}^{(0)}, \Psi_{\mathrm{id}}^{(0)})} {p_{\mathrm{fail}}},	\quad F_{\mathrm{rec}}^{\mathrm{succ}}= \frac{  \sum_{i=1}^N p_i F(\Psi_{\mathrm{out}}^{(i)}, \Psi_{\mathrm{id}}^{(i)})}{p_{\mathrm{succ}}}.
\end{equation}
One has to observe that if the failure is encountered in the first round it is no longer a  two-step protocol, since the teleportation of the first subsystem failed. Nonetheless it is still valuable to study the conditional fidelity after failure, in order to know if the remaining resource is suitable for at least another round of the teleportation. On the other hand, this quantity doed not influence the total recycling fidelity since it is known \cite{ishizaka_quantum_2009, Studzinski2017, majenz2} that $p_0$ vanishes in the limit of large $N$.

Note, that two quantities $F_\mathrm{ent}$ and $F_\mathrm{rec}$ are related indirectly. The low discrepancy between the resource after the teleportation of half of an EPR state and the ideal resource suited for subsequent PBT implementation means that the faithful teleportation of two maximally entangled states is possible, in consequence meaning that such approach provides a faithful teleportation channel for any two-quantum states (thanks to the connection between the entanglement fidelity and average teleportation fidelity).

In this paper the above reasoning is employed to study the reapplication of another probabilistic protocol to the already used resource. It is especially important in the case of probabilistic scheme with EPR resource, since the measurement in this scenario is not a PGM so the results quantifying $F_{\mathrm{ent}}$ and $p_{\mathrm{succ}}$ do not apply there.

\section{Summary of results}

The study of entanglement fidelity $F_{\mathrm{ent}}(\mathcal{N})$ for two-step PBT channel is the main result of this paper, together with the probability of success in the case of repetitive application of probabilistic measurements. This is the novel result that quantifies explicitly the two measures describing the performance of the teleportation channel in regard to repetitive application of PBT protocol.

\subsection{The quality of PGM-based two-step PBT}

The main result of this paper is a full description of both figures of merit (entanglement fidelity and average probability of success) in the case of two step PBT employing PGM measurements. The results are summed up by the following Theorems.

\begin{theorem}
	\label{thm:ent_fid_inf}
	The unnormalised entanglement fidelity of PGM-based two-step PBT channel $\mathcal{N}$ is given by
	\begin{equation}
		\label{eq:ent_fid_inf}
		F_\ent(\mathcal{N}) = \frac{1}{d^4} v\tp Mv,
	\end{equation}
	where $M$ is an analogue of teleportation-matrix~\cite{Mozrzymas2021optimalmultiport}, and the vector $v$ depends on the preparation procedure on Alice's side.
\end{theorem}

Since in deterministic inexact scheme the channel is trace-preserving no normalisation is required and thus one can optimise the resource state in order to suit two-step scenario, meaning that the entanglement fidelity reduces to the maximal eigenvalue of $M$.

\begin{theorem}[\cref{thm:two_step_psucc}, informal]
	\label{thm:p_succ_inf}
	The average probability of success in the PGM-based two-step probabilistic PBT coincides with the average success probability of the optimal multi-port based probabilistic scheme.
\end{theorem}

\subsection{Resource degradation results}

In this paper we also employ this approach to study the probabilistic regime. It enabled the study of two-step PBT with EPR resource, which is not quantified by the main result of this paper, since the POVM employed is not a PGM.

\begin{theorem}[$F_{\mathrm{rec}}$ for EPR resource]
	\label{thm:rec_fid_inf_std}
	After one round of probabilistic PBT with EPR resource, the degradation of the resource is small for big $N$, enabling the implementation of second round of probabilistic protocol, implying that two-step PBT is valid for other POVMs than a PGM.
\end{theorem}

\begin{theorem}[$F_{\mathrm{rec}}$ for optimised resource]
	\label{thm:rec_fid_inf_opt}
	After one round of probabilistic PBT with optimised resource, the degradation of the resource does not decrease for large N, meaning that the resource is significantly different from the original one.
\end{theorem}


%
%

\section{Mathematical preliminaries}

\subsection{Representation theory of the partially transposed permutation algebra}

The setting of our problem is naturally suited for using tools from mixed Schur--Weyl duality for the matrix algebra of partially transposed permutations.
Here we present a summary of the necessary results, and refer the reader to \cite{grinko2023gelfandtsetlin} for more details.
Some aspects of the representation theory of partially transposed permutation matrix algebras were also studied before in \cite{Stu1, Moz1, MozJPA}.

The matrix algebra $\A^d_{N,2}$ of \emph{partially transposed permutations} acts on $N+2$ qudits, each of local dimension $d$.
Its generators $\sigma_1, \dotsc, \sigma_{N+1}$ act on $(\C^d)\xp{N+2}$ in the following way:
for all $x_1, \dotsc, x_{N+2} \in [d] \defeq \set{1,\dotsc,d}$,
\begin{equation}
	\label{eq:Brauer action}
	\sigma_i \, \ket{x_1,\dotsc,x_{N+2}}
	\defeq
	\begin{cases}
		\ket{x_1,\dotsc,x_{i+1},x_i,\dotsc,x_{N+2}}, &
		\text{$i \neq N$},                             \\
		\ket{x_1,\dotsc,x_{N-1}}
		\x
		(\delta_{x_N,x_{N+1}}
		\sum_{k=1}^d \ket{k,k})
		\x
		\ket{x_{N+2}}
		,                                            &
		\text{$i = N$}.
	\end{cases}
\end{equation}
In other words, $\sigma_i$ with $i \neq N$ are \emph{transpositions} that exchange qudits $i$ and $i+1$, while $\sigma_N$ is a \emph{contraction} that projects qudits $N$ and $N+1$ on the un-normalized maximally entangled state.
The irreducible representations or \emph{irreps} of $\A_{N,2}^d$ are labeled by pairs of Young diagrams of the following four types:
\begin{equation}
	\label{eq:IrrA}
	\Irr{\A_{N,2}^{d}} \defeq
	\Set[\Big]{
		(\nu,\0) \given
		\nu \pt_d N - 2} \sqcup
	\Set[\Big]{
		(\lambda,\yd[0.5]{1}) \given
		\lambda \pt_{d-1} N - 1} \sqcup
	\Set[\Big]{
		(\mu,\yd[0.35]{1,1}) \given
		\mu \pt_{d-2} N}\sqcup
	\Set[\Big]{
		(\mu,\yd[0.35]{2}) \given
		\mu \pt_{d-1} N},
\end{equation}
where $\nu \pt_d N - 2$ means that $\nu$ is a partition of $N-2$ with at most $d$ non-zero parts, $\0$ denotes the empty diagram and $\square = (1)$.

More generally, the irreducible representations of $\A_{N,M}^d$ for arbitrary $N,M \geq 0$ are labeled by pairs of partitions $\Lambda = (\lambda,\lambda')$ such that $\ell(\lambda) + \ell(\lambda') \leq d$ and $\lambda \pt N-k$ and $\lambda' \pt M-k$ for some $0 \leq k \leq \min(N,M)$ \cite{grinko2023gelfandtsetlin} (in our case $M=2$ and $k=0,1,2$).
We slightly abuse notation by not including $d$ as part of $\Lambda$ since $d$ is assumed to be fixed throughout.
However, knowing $d$ is necessary to unambiguously convert $\Lambda$ into a \emph{staircase} of length $d$, which is another convenient way of labeling the irreps of $\A_{N,M}^d$ \cite{stembridge1987rational,grinko2023gelfandtsetlin}.
When $M=0$, $\A_{N,0}^d$ coincides with the usual tensor representation of the \emph{symmetric group} $\S_N$ on $(\C^d)^{\otimes N}$. Therefore, up to level $N$ when $M=0$, we can omit for brevity the second diagram $\lambda'$ from $\Lambda$ and only refer to $\Lambda$ by its first Young diagram $\lambda$.

Mixed Schur--Weyl duality says that the tensor representation $(\C^d)^{\otimes N + M}$ of $\A_{N,M}^d$ decomposes as follows into irreducible representations:
\begin{equation}
	\label{eq:mixed_schur_weyl}
	(\C^d)^{\otimes N + M} \cong \bigoplus_{\Lambda \in \Irr{\A_{N,M}^{d}}} \mathcal H_\Lambda \otimes \mathcal W_\Lambda,
\end{equation}
where $\mathcal{H}_\Lambda$ is a irrep of $\A_{N,M}^d$, and $\mathcal{W}_\Lambda$ is an irrep of the unitary group $\U{d}$ with highest weight $\Lambda$ (the second Young diagram $\lambda'$ of $\Lambda$ describes the negative part of the highest weight $\Lambda$).

It is convenient to describe the representation theory of $\A_{N,2}^d$ using the \emph{Bratteli diagram} $\Brat$
for the sequence of algebras
$\A_{0,0}^d \hookrightarrow \A_{1,0}^d \hookrightarrow \cdots \hookrightarrow \A_{N,0}^d \hookrightarrow \A_{N,1}^d \hookrightarrow \A_{N,2}^d$,
which is a certain directed acyclic simple graph \cite{grinko2023gelfandtsetlin}.
The vertices of $\Brat$ are divided into $N+3$ \emph{levels} denoted by $i = 0,\dotsc,N+2$. These levels correspond to sets of irreducible representations $\Irr{\A_{0,0}^d},\Irr{\A_{1,0}^d},\dotsc,\Irr{\A_{N,2}^d}$ of the corresponding algebras $\A_{0,0}^d,\A_{1,0}^d,\dotsc,\A_{N,2}^d$.
The vertices at level $i \in \set{0,\dotsc,N}$ are labeled by $\Lambda = (\lambda, \0)$ or simply $\lambda$, where $\lambda \pt_d i$ is a Young diagram with $i$ cells, while the vertices at the last two levels $N+1$ and $N+2$ are labeled by pairs of Young diagrams $\Lambda$ corresponding to the irreps of $\Irr{\A_{N,1}^{d}}$ and $\Irr{\A_{N,2}^{d}}$, see \cref{eq:IrrA}.

When $i \in \set{0,\dotsc,N}$, the vertices $\mu \, \pt \, i-1$ and $\lambda \, \pt \, i$ are connected, denoted as $\mu \rightarrow \lambda$, if $\lambda$ can be obtained from $\mu$ by adding a cell, i.e., $\lambda = \mu \cup a$ for some $a \in \AC_d(\mu)$ where $\AC_d(\mu)$ denotes the set of addable cells of $\mu$.
In other words, $\mu \cup a$ must be a valid Young diagram with at most $d$ rows.
Furthermore, for levels $\geq N$ vertices at consecutive levels are connected by either adding a cell to the right Young diagram $\lambda'$ of $\Lambda = (\lambda,\lambda')$ or by removing a cell from the left diagram $\lambda$.
The Bratteli diagram $\Brat$ consists of all vertices from all levels and the directed edges between them. 
We denote the only vertex at level $0$ by $\0$ and call it \emph{root}, while the vertices at the last level of $\Brat$ we call \emph{leaves} (they correspond bijectively to the irreps $\Irr{\A_{N,2}^{d}}$).
For any leaf $\Lambda \in \Irr{\A_{N,2}^{d}}$, we denote by
\begin{equation}
	\Paths(\Lambda,\Brat) \defeq
	\Set[\Big]{T = (T^0 \to T^1 \to \dotsc \to T^{N+1} \to T^{N+2}) = (T^0,T^1,\dotsc,T^{N+1},T^{N+2})}
	\label{eq:paths}
\end{equation}
the set of all paths in $\Brat$ starting at the root $T^0 = \0$ and terminating at $\Lambda$.
Similar to \cref{eq:paths}, for any intermediate vertex $\lambda \pt k$ at level $k \leq N$ in $\Brat$, we use $\Paths_{k}(\lambda,\Brat)$ to denote the set of all paths in $\Brat$ terminating at $\lambda$.\footnote{Sometimes, it will be also convenient to abuse the notation by dropping $\Brat$ and simply writing $\Paths_{k}(\lambda)$ if it is clear from the context what is the underlying Bratteli diagram.}
Furthermore, we denote by $\Paths(\Brat)$ the set of all paths in $\Brat$, i.e.,
$\Paths(\Brat) \defeq \bigsqcup_{\Lambda \in \Irr{\A_{N,2}^{d}}} \Paths(\Lambda,\Brat)$.

For any path $T = (T^0,\dotsc,T^{N+2}) \in \Paths(\Lambda,\Brat)$ and $i \in [N+2]$, we define the \emph{walled content} of $i$ in $T$ as
\begin{equation}
	\wcont_i(T) \defeq
	\begin{cases}
		\cont(T^{i}_l \setminus T^{i-1}_l)   & \text{if } i \leq N,                             \\
		\cont(T^{i}_r \setminus T^{i-1}_r)+d & \text{if } i > N \text{ and } T^{i}_l=T^{i-1}_l, \\
		-\cont(T^{i-1}_l \setminus T^{i}_l)  & \text{if } i > N \text{ and } T^{i}_r=T^{i-1}_r,
	\end{cases}
	\label{eq:conti}
\end{equation}
where \emph{content} of a cell $c=(i,j)$ is given by $\cont(c) \defeq j - i$.
The \emph{axial distance} between $i$ and $i+1$ in $T$ is
\begin{equation}
	r_i(T) \defeq \wcont_{i+1}(T)-\wcont_{i}(T).
	\label{eq:ri}
\end{equation}

For a given $\Lambda \in \Irr{\A_{N,2}^{d}}$, the corresponding irrep $\psi_\Lambda$ of $\A_{N,2}^d$ has a convenient explicit description in the so-called \emph{Gelfand--Tsetlin basis}
$\set{\ket{T} \mid T \in \Paths(\Lambda,\Brat)}$.
Recall from \cref{eq:Brauer action} that $\A_{N,2}^{d}$ is generated by $N$ transpositions $\sigma_1,\dotsc,\sigma_{N-1},\sigma_{N+1}$ and one contraction $\sigma_N$.
For any irrep $\Lambda \in \Irr{\A_{N,2}^{d}}$, the transposition $\sigma_i$ acts on a given path $T \in \Paths(\Lambda,\Brat)$ according to the so-called Young--Yamanouchi formula:
\begin{align}
	\psi_\Lambda(\sigma_i) \, \ket{T}
	 & \defeq \frac{1}{r_i(T)} \, \ket{T} + \sqrt{1 - \frac{1}{r_i(T)^2}} \, \ket{\sigma_i T}
	\qquad
	\text{for $i \neq N$},
	\label{gtbasis:transpositions}
\end{align}
where $\sigma_i T$ denotes the path $T$ with vertex $T^i$ at level $i$ replaced by $T^{i-1} \cup (T^{i+1} \setminus T^i)$, which corresponds to adding cells $i$ and $i+1$ in the opposite order.
The action of $\sigma_N$ is given by (see \cite[Theorem 3.2]{grinko2023gelfandtsetlin}) for every irrep $\Lambda$ and a valid path $S \in \Paths_{N-1}((\lambda,\0),\Brat)$:
\begin{align}
	\psi_{\Lambda}(\sigma_N) \, \ket{S\to \lambda\cup a \to (\lambda,\0) \to \Lambda}
	 & \defeq \sum_{a'\in \AC_d(\lambda)} \frac{ \sqrt{ m_{\lambda\cup a} m_{\lambda\cup a'}}}{ m_{\lambda} } \ket{S\to \lambda\cup a' \to (\lambda,\0) \to \Lambda},
	\label{gtbasis:contraction1}
\end{align}
where $m_\lambda$ denotes the dimension of the $\lambda$-irrep of the unitary group $\U{d}$.
According to the well-known \emph{Weyl dimension} formula, the dimension of unitary group irrep labeled by a general staircase $\Lambda$ is given by
\begin{equation}
	m_{\Lambda} \defeq \prod_{1 \leq i < j \leq d}
	\frac{\Lambda_i - \Lambda_j + j - i}{j-i}.
\end{equation}

Finally, we need the notion of \emph{matrix units} of the partially transposed permutation matrix algebra $\A_{N,M}^d$. For a given irrep $\Lambda \in \Irr{\A_{N,M}^d}$ and a pair of paths $T,S \in \Paths(\Lambda,\Brat)$ we define the matrix unit $E_{T,S}$ as a matrix in the standard basis, which looks like a matrix unit $\ket{T}\bra{S}$ in the Gelfand--Tsetlin basis, i.e.
\begin{equation}
	U_{\mathrm{Sch}(N,M)} E_{T,S} U_{\mathrm{Sch}(N,M)}^{\dagger} \defeq  \ket{T}\bra{S} \otimes I_\Lambda,
\end{equation}
where $U_{\mathrm{Sch}(N,M)}$ denotes mixed Schur transform, i.e. an unitary transformation between computational and the mixed Schur basis, see~\cite{grinko2023gelfandtsetlin}. In the context of mixed Schur--Weyl decomposition in \cref{eq:mixed_schur_weyl}, $\ket{T}\bra{S}$ acts on the irrep $\mathcal H_\Lambda$, and $I_\Lambda$ acts on the unitary group irrep $\mathcal W_\Lambda$.

Finally, there is a very useful formula for the partial trace over the matrix units for the Gelfand--Tsetlin basis:
\begin{lemma}[\cite{ram1992matrix}]
	\label{lem:partial_trace_E}
	Consider matrix units $E_{T,S}$ on $N$ qudits in the Gelfand--Tsetlin basis, where $T,S \in \Paths_N(\Lambda)$ and $T_{N-1} = \nu$ and $S_{N-1} = \nu'$. Then
	\begin{equation}
		\Tr_N E_{T,S} = \delta_{\nu,\nu'} \frac{m_\Lambda}{m_\nu} E_{\bar{T},\bar{S}},
	\end{equation}
	where $\bar{T}$ denotes the truncated path $T$ by its last vertex $\Lambda$.
\end{lemma}

\subsection{Optimal measurements and resource states for pPBT}
\label{ssec:PBT}
Since the building blocks of two-step PBT are measurements and resources employed in single round PBT, they are described below.

Depending on the shared resource state and the measurements applied by Alice, one can distinguish two types of pPBT protocols: standard and optimized.

\subsubsection{Standard pPBT}
In the standard scenario, the shared state consists of $N$ \emph{EPR pairs} $\ket{\psi^+} \defeq \frac{1}{ \sqrt{d} }\sum_{k=1}^{d} \ket{k,k}$:
\begin{equation}
	\label{eq:standard_pPBR_resource}
	\ket{\Psi^+}_{AB} \defeq	\bigotimes_{i=1}^{N} \ket{\psi^+}_{A_i B_i}.
\end{equation}
Due to problem's symmetries~\cite{ishizaka_asymptotic_2008}, Alice's measurement $\Pi = \{\Pi_i\}_{i=0}^N$ on $\bar{A}_1 A$ without loss of generality takes the form
\begin{equation}
	\label{eq:standard_pPBT_measurement}
	\Pi_i = \psi^+_{A_i \bar A_1}\otimes\Theta_{i^c}, \qquad
	\Theta_{i^c} = \sum_{\lambda \vdash_d N-1} \frac{d}{\gamma(\lambda)} P^{i^c}_{\lambda}, \quad i=1,\dotsc,N, \qquad
	\Pi_0 = I - \sum_{ i=1}^{N} \Pi_i
\end{equation}
where notation $\Theta_{i^c}$
means that the operator $\Theta$ acts everywhere but $A_i$, $\psi^+_{A_i \bar A_1}$ is the density matrix of the maximally entangled state $\ket{\psi^+}$ on subsystems $A_i$ and $\bar A_1$, and
\begin{align}
	\label{eq:gamma_star}
	\gamma(\lambda) & \defeq \max_{\substack{a\in \AC(\lambda) }}\gamma_a(\lambda) = d + \lambda_1, \qquad \gamma_a(\lambda) \defeq d+\cont{(a)},
\end{align}
and
\begin{equation}
	\label{eq:P_alpha}
	P_{\lambda} \defeq \sum_{S\in\Paths_{N-1}{(\lambda)}} U_{\textrm{Sch}(N-1)}\ct \ketbra{S}{S} U_{\textrm{Sch}(N-1)}
\end{equation}
is the projector onto irrep $\lambda$ in $\left(\mathbb{C}^d\right)^{\otimes {N-1}}$, while the superscript $i^c$ in $P_{\lambda}^{i^c}$ means it is supported on systems $A_1,\dots,A_{i-1},A_{i+1},\dots,A_N$.
The effects $\Pi_i$ for $i=1,\dots,N$ correspond to successful teleportation to the $i$-th port, while the remaining effect $\Pi_0$ corresponds to failure.

\subsubsection{Optimal pPBT}
When Alice optimizes the resource state, it is of the form \cite{ishizaka_quantum_2009, StuNJP}
\begin{equation}
	\label{eq:opt_pPBT_resource}
	\ket{\Psi}_{AB} \defeq \bigl( O_A\otimes I_{B} \bigr) \ket{\Psi^+}_{AB}, \qquad
	O_A \defeq \sum_{\mu \vdash_d N} \sqrt{c_{\mu}} P_{\mu},
\end{equation}
where we introduce the following definitions for a given $\mu \pt_d N$:
\begin{equation}
	\label{eq:c_mu_pPBT}
	c_\mu \defeq \frac{d^N g(N) m_{\mu}}{d_\mu}, \qquad
	g(N) \defeq \frac{1}{\sum_{\mu \pt_d N} m_{\mu}^2} = \frac{1}{\binom{N+d^2-1}{d^2-1}}, \qquad
	f_\mu \defeq \frac{c_\mu d_\mu m_\mu}{d^N} = \frac{ m_\mu^2}{\sum_{\nu \pt_d N} m_\nu^2}.
\end{equation}
The optimal pPBT measurement on $A \bar A_1 = (\C^d)^{\otimes (N+1)}$ for the optimized resource state is a \emph{Pretty Good Measurement} (PGM) $\Pi$ given by \cite{Studzinski2017}
\begin{align}
	\label{def:PGM_PBT}
	\Pi_i \defeq \rho^{-1/2} \rho_i \rho^{-1/2}, \qquad
	\rho \defeq \sum_{i=1}^N \rho_i, \qquad
	\rho_i \defeq (i,N) \sigma_N (i,N), \qquad
	\Pi_0 \defeq I - \sum_{i=1}^N \Pi_i,
\end{align}
where $\sigma_N$ denotes the contraction between systems $A_N$ and $\bar A_1$, see \cref{eq:Brauer action}, and $(i,N)$ is the transposition that swaps systems $i$ and $N$, which acts on $(\C^d)^{\otimes (N+1)}$ as the tensor representation \eqref{eq:Brauer action} of $\S_N$.
Notice that $\rho$ commutes with the action of $\S_N$, hence the POVM elements $\Pi_i$ for $i \in [N]$ can be written as
\begin{equation}
	\label{eq:group_cov}
	\Pi_i = (i,N) \Pi_N (i,N),
\end{equation}
meaning that $\Pi$ is group-covariant \cite{Decker2004} under the cyclic group on $N$ elements.

We now present a useful description of the square root measurement \eqref{def:PGM_PBT} in the Gelfand--Tsetlin basis \cite{grinko2023gelfandtsetlin}. For each irrep $\Lambda$, we denote by $\Pi_k^\Lambda$ an operator $\Pi_k$ restricted to the irrep $\Lambda$:
\begin{equation}
	U_{\mathrm{Sch}(N,1)} \Pi_k U_{\mathrm{Sch}(N,1)}^{\dagger} = \bigoplus_{\Lambda \in \Irr{\A_{N,1}^{d}}} \Pi_k^\Lambda \otimes I_{m_{\Lambda}},
\end{equation}
where $I_{m_{\Lambda}}$ denotes the identity matrix on the corresponding unitary group irrep register. In fact, we are interested mostly in $ \Lambda = (\lambda,\0)$ where $\lambda \pt_{d} N-1$ type, since the expressions describing our figures of merit are only supported on that types of irreps.

The representations of $\sigma_N$ and $\rho^{-1/2}$ as elements of $\A_{N,1}^d$ are known explicitly \cite{Grinko2024}:
\begin{equation}
	\label{eq:sigma_n_Lambda}
	\sigma_N^\Lambda =
	\begin{cases}
		\displaystyle \sum_{S\in \Paths (\lambda )} 
		\ketbra{v_{S,\lambda}}{v_{S,\lambda}} &
		\text{if $\Lambda = (\lambda,\0)$},     \\[15pt]
		\displaystyle 0                       &
		\text{if $\Lambda = (\mu,\square)$},
	\end{cases}
\end{equation}
where $\ket{v_{S,\lambda}} = \sum_{ a\in \AC(\lambda)}^{} \sqrt\frac{ { m_{\lambda \cup a}} }{  m_\lambda} \ket{S\rightarrow\lambda\cup a\rightarrow (\lambda, \0) }$
and
\begin{equation}
	\label{eq:rho_Lambda}
	(\rho^{-1/2})^\Lambda =
	\begin{cases}
		\displaystyle \sum_{a \in \AC_d(\lambda )} \sum_{T\in \Paths_N (\lambda \cup a)}  \frac{1}{\sqrt{d+\cont{(a)}}} \proj{T\to(\lambda,\0)} &
		\text{if $\Lambda = (\lambda,\0)$},                                                                                                       \\[15pt]
		\displaystyle 0                                                                                                                         &
		\text{if $\Lambda = (\mu,\square)$}.
	\end{cases}
\end{equation}

\begin{lemma}
	\label{lem:PGM_new_form}
	The POVM element $\Pi_N$ can be expressed as
	\begin{equation}
		\label{eq:PI_N}
		\displaystyle
		\Pi_N =
		\sum_{\lambda\pt_d N-1}
		\sum_{S\in \Paths_{N-1} (\lambda )}
		G_{\lambda,S} \,\sigma_N \, G_{\lambda,S}
	\end{equation}
	where
	\begin{equation}
		U_{\mathrm{Sch}(N,1)}^{\dagger} G_{\lambda,S} U_{\mathrm{Sch}(N,1)} = \sum_{a \in \AC_d (\lambda )} \frac{1}{\sqrt{d+\cont{(a)}}} E_{S \to \lambda\cup a,S \to \lambda\cup a} \otimes I_d.
	\end{equation}
\end{lemma}
\begin{proof}
	This follows from a direct calculation using \cref{eq:sigma_n_Lambda,eq:rho_Lambda}. Notice that although matrices $G_{\lambda,S}$ have non-trivial support on irreps $\Gamma = (\mu,\yd[0.5]{1})$ at level $N+1$, $\sigma_N^\Lambda$ is supported only on irreps $ \Lambda = (\lambda,\0)$ at level $N+1$ of the Bratteli diagram for $\A_{N,1}^d$.
\end{proof}

Finally, the following easy lemma \cite{grinko2023gelfandtsetlin} is useful for our calculations:
\begin{lemma}
	For any partition $\lambda = (\lambda_1,\dotsc,\lambda_d) \pt N-1$ and $a \in \AC(\lambda)$,
	\begin{equation}
		\label{eq:content}
		N \frac{d_\lambda}{m_\lambda} \frac{m_{\lambda\cup a}}{d_{\lambda\cup a}} = d +\cont{(a)}.
	\end{equation}
\end{lemma}



\section{Performance of the two-step PBT}
\label{sec:two-step_pbt}

The succeeding performance of successful executions of the protocol (i.e.~such that neither the first, nor the second measurement outcome indicates failure) can be viewed as a protocol that teleports two quantum states.
Thus, the examination of the quantities describing its performance give information about potential usefulness of the resource state that was employed to succesfully teleport one quantum state, to teleport another one.

Recall that $A = A_1 \dots A_N$ and $B = B_1 \dots B_N$ denote Alice's and Bob's registers that store the shared optimized resource state
\begin{equation}
	\ket{\Psi}_{AB} \defeq (O_A \otimes I_B) \ket{\Psi^+}_{AB},
\end{equation}
where $\ket{\Psi^+}_{AB}$ consists of $N$ qudit EPR pairs, see \cref{eq:standard_pPBR_resource}, and the operator $O_A$ is used to optimize the state by adjusting its Schmidt decomposition.
We will denote the density matrix of $\ket{\Psi}_{AB}$ by $\Psi_{AB}$.
In addition to $A$, Alice also has a register $\bar A = \bar A_1 \bar A_2$ that stores the two-qudit state $\rho_{\bar A}$ that she wants to teleport to Bob in two consecutive PBT rounds, where $\bar A_1$ and $\bar A_2$ must be teleported in rounds 1 and 2, respectively.

As indicated in \cref{fig:protocol}, Alice first performs a POVM $\Pi = \set{\Pi_0, \dotsc, \Pi_N}$ on her registers $A \bar A_1$. If her outcome is $i \neq 0$, the first round of PBT succeeds and their joint post-measurement state becomes
\begin{equation}
	\Psi^{(i)}_{\mathrm{out}} \defeq \frac{\Tr_{A_i B_i \bar A_1} \sof[\big]{\sqrt{\Pi_i} (\Psi_{AB} \otimes \rho_{\bar A}) \sqrt{\Pi_i}}}{\Tr\sof[\big]{\Pi_i (\Psi_{AB} \otimes \rho_{\bar{A}})}}.
\end{equation}
Alice then performs a second POVM $\widetilde{\Pi} = \set{\widetilde{\Pi}_0, \dotsc, \widetilde{\Pi}_{N-1}}$ on her remaining registers $A_i^c \bar A_2$, where $A_i^c \defeq A \setminus A_i$ denotes the complement of $A_i$ in $A$.
The second round succeeds if she obtains an outcome $j \neq 0$.

This two-step process can be described by the following completely positive map $\mathcal{N}$ from $\bar A_1 \bar A_2$ to $\bar B_1 \bar B_2$:
\begin{equation}
	\begin{split}
		\label{eq:channel_N}
		\mathcal{N}(\rho_{\bar A}) \defeq &
		\sum_{i = 1}^N \sum_{\substack{j = 1 \\ j \neq i}}^{N}
		\Tr_{A B^c_{i,j} \bar A}
		\sof*{
			\sqrt{\widetilde{\Pi}_{j}} {\sqrt{\Pi_i}} ( \Psi_{AB} \otimes \rho_{\bar A} ) {\sqrt{\Pi_i}} \sqrt{\widetilde{\Pi}_{j}}
		}_{B_i,B_j \to \bar B_1, \bar B_2}
	\end{split}
\end{equation}
where $B^c_{i,j} \defeq B \setminus \set{B_i, B_j}$, which describes the post-measurement state in case of success ($i,j \neq 0$) in both rounds of the protocol.
Note that $\mathcal{N}$ is generally trace-decreasing (i.e., it is a subchannel) since we have omitted the POVM elements $\Pi_0$ and $\widetilde{\Pi}_0$ which indicate failure in the first and second round, respectively.

The success probability of the two-step protocol is
\begin{equation}
	\label{def:prob_suc}
	p_{\mathrm{succ}} \defeq \Tr \sof*{\mathcal{N}(I/d^2)},
\end{equation}
and the corresponding entanglement fidelity is
\begin{equation}
	\label{def:entanglement_fidelity}
	F_e(\mathcal{N}) \defeq \Tr
	\sof*{
	\psi^+_{\bar{B}_1R_1} \otimes
	\psi^+_{\bar{B}_2R_2}
	\of*{\mathcal{N}_{\bar{A}_1\bar{A}_2\to \bar{B}_1\bar{B}_2} \otimes I_{R_1R_2}}
	\of*{\psi^+_{\bar A_1R_1}\otimes \psi^+_{\bar A_2R_2}}
	}
\end{equation}
where $R_1,R_2$ are auxiliary registers, and $\psi^+_{XY}$ are maximally entangled states shared between $X$ and $Y$ registers.

Since the measurements $\Pi,\widetilde\Pi$ are covariant under the action of $\S_N$, $\S_{N-1}$, respectively,
\begin{align}
	R(\sigma) \Pi_i R(\sigma)^\dagger              & = \Pi_{\sigma(i)},\quad \sigma \in \S_N,                  \\
	R(\sigma') \widetilde \Pi_j R(\sigma')^\dagger & = \widetilde \Pi_{\sigma'(j)},\quad \sigma' \in \S_{N-1}.
\end{align}

Calculating $F_e(\mathcal{N})$ can be interpreted as contracting a tensor network obtained from the circuit in \cref{fig:protocol}.
In particular, it is easy to see that by bending wires of \cref{fig:protocol}, $F_e(\mathcal{N})$ can be transformed into taking trace of the operator as in \cref{fig:fidelity}.
Therefore, the entanglement fidelity of the channel \eqref{eq:channel_N} reads
\begin{equation}
	\label{eq:fidelity_formula}
	F_e(\mathcal{N})=  \frac{ N(N-1)}{d^{N+4}} \Tr \sof*{ O_A \sqrt{\Pi_N} \widetilde{\Pi}_{N-1} \; \sqrt{\Pi_N} \; O_A \; \tau_2 }
\end{equation}
where $\Pi_N, \widetilde{\Pi}_{N-1}$ are POVMs elements, $O_A$ corresponds to the preparation of the optimal state, and
\begin{equation}
	\label{eq:tau_2}
	\tau_2 \defeq d^2(I_{A_1,\dots,A_{N-2}} \otimes {\psi^+}_{A_{N} \bar A_1}  \otimes{\psi^+}_{A_{N-1} \bar A_2})
\end{equation}
is an unnormalized projector onto two EPR pairs.
In the following, it makes sense to call registers $\bar A_1, \bar A_2$ as $A_{N+1}, A_{N+2}$, see \cref{fig:fidelity}.

\begin{figure}[ht!]
	\centering
	$F_e(\mathcal{N}) = \frac{ N(N-1)}{d^{N+4}} \Tr \sof*{ \adjincludegraphics[valign=c,width=0.6\linewidth,page=4]{./figs/protocol_figures.pdf} }$
	\caption{Entanglement fidelity of the channel $\mathcal{N}$ can written as a tensor network contraction by bending and rearranging the wires and tensors from \cref{fig:protocol}.}
	\label{fig:fidelity}
\end{figure}

\begin{theorem}\label{thm:two_step_fidelity}
	The entanglement fidelity of two-step PBT with $N$ ports of local dimension $d$ can be expressed as
	\begin{align}
		\label{eq:F_e_answer}
		F_e(\mathcal{N})
		= \frac{1}{d^4} v\tp M v,
	\end{align}
	where the components of vector $v$
	are labeled by partitions $\mu \pt_d N$:
	\begin{align}
		v_\mu & \defeq \sqrt{f_\mu},
	\end{align}
	and the matrix  $M\in \mathbb{R}^{p(N,d)\times p(N,d)}$, where $p(N,d)$ denotes the number of partitions of $N$ with length at most $d$, is defined as
	$M \defeq \sum_{\nu \pt_d N-2} M^\nu$
	where
	$M^\nu \defeq X^\nu{\tp} X^\nu$, and
	$X^\nu \defeq \frac{1}{d^2} H^\nu S^\nu$ where $S^\nu \in \mathbb{R}^{|\AC_d(\nu)| \times p(N,d)}$ and $H^\nu \in \mathbb{R}^{|\AC_d(\nu)| \times |\AC_d(\nu)|}$ are defined as:
	\begin{align}
		H^\nu_{a,a'}    & \defeq
		\begin{dcases}
			\sum_{b\in \AC_d(\nu \cup a)} \sqrt{\frac{d + \cont(b)}{d + \cont(a)}} \frac{q(b|\nu \cup a)}{(\cont(a)-\cont(b))^2} &
			\text{if $a=a'\in \AC_d(\nu)$},                                                                                        \\
			\sqrt{q(a|\nu \cup a')q(a'|\nu \cup a)} \of*{1- \frac{1}{(\cont(a)-\cont(a'))^2}}                                    &
			\text{if $a\neq a'\in \AC_d(\nu)$},
		\end{dcases}                              \\
		S^{\nu}_{a,\mu} & \defeq \delta_{\mu,\AC_d(\nu \cup a)} \sqrt{\frac{1}{q(a|\nu)}} \sqrt{\frac{1}{\sum_{b \in \AC_d(\nu \cup a)} q(b|\nu \cup a)}} , \\
		q(a | \lambda)  & \defeq \frac{d_{\lambda\cup a}}{N \, d_\lambda},
	\end{align}
	where the delta function is defined as $\delta_{\mu,\AC_d(\nu \cup a)} \defeq 1$ iff there exists $b \in \AC_d(\nu \cup a)$ such that $\mu = \nu \cup a \cup b$, and $\delta_{\mu,\AC_d(\nu \cup a)} \defeq 0$ otherwise. Note that for every $\lambda$, $q(a | \lambda)$ is a measure over all addable cells $\AC(\lambda)$, i.e.\ $\sum_{a \in \AC(\lambda)} q(a | \lambda) = 1$.
\end{theorem}
The proof of this theorem is located in Appendix~\ref{apx:proof}. Now, we compute $p_{\mathrm{succ}}$. We get the following general formula, which depends on the squared amplitudes $\set{f_\mu}_{\mu \pt_d N}$ of the resource state:
\begin{theorem}
	\label{thm:two_step_psucc}
	The average probability of succes in the two-step PBT protocol with arbitrary resource state specified by $\set{f_\mu}_{\mu \pt_d N}$ is given by
	\begin{align}
		p_{\mathrm{succ}} = \frac{1}{d^2} \sum_{\lambda \pt_d N-1} \of*{\sum_{a \in \AC_d(\lambda)} \frac{f_{\lambda \cup a}}{m_{\lambda \cup a}}} \of*{\sum_{r \in \RC(\lambda)} m_{\lambda \setminus r} }.
	\end{align}
	In particular, for our two-step pPBT protocol with specific $f_\mu = \frac{ m_\mu^2}{\sum_{\nu \pt_d N} m_\nu^2}$ the above formula gives the following simple expression:
	\begin{equation}
		p_{\mathrm{succ}} = \frac{N (N-1)}{(N+d^2-1)(N+d^2-2)}.
	\end{equation}
\end{theorem}
Notice that $p_{\mathrm{succ}}$ coincides with the success probability of the probabilistic multi-PBT (see Theorem 3 in \cite{Mozrzymas2021optimalmultiport}).

\begin{figure}[!ht]
	\centering
	\begin{subfigure}{0.45\textwidth}
		\begin{tikzpicture}
\begin{axis}[
    width=\linewidth,
    height=0.7\linewidth,
    xmin=0, xmax=200,
    ymin=0, ymax=1.001,
    xlabel={$N$},
    ylabel={$p_{\mathrm{succ}}$},
    grid=both,
    legend style={at={(1.02,0.5)},anchor=west},
    legend cell align=left,
    tick label style={/pgf/number format/fixed}
]

  \addplot[
    only marks,
    mark=*,
    mark size=1.2pt,
    color=blue,
    mark options={fill=blue}
  ] coordinates {
    (5, 5/14) (10, 15/26) (15, 35/51) (20, 190/253) (25, 50/63)
    (30, 145/176) (35, 595/703) (40, 260/301) (45, 165/188)
    (50, 1225/1378) (55, 495/551) (60, 590/651) (65, 1040/1139)
    (70, 805/876) (75, 925/1001) (80, 3160/3403) (85, 595/638)
    (90, 1335/1426) (95, 4465/4753) (100, 1650/1751) (105, 910/963)
    (110, 5995/6328) (115, 2185/2301) (120, 2380/2501)
    (125, 3875/4064) (130, 2795/2926) (135, 3015/3151)
    (140, 9730/10153) (145, 1740/1813) (150, 3725/3876)
    (155, 11935/12403) (160, 4240/4401) (165, 2255/2338)
    (170, 14365/14878) (175, 5075/5251) (180, 5370/5551)
    (185, 8510/8789) (190, 5985/6176) (195, 6305/6501)
    (200, 19900/20503)
  };

  \addplot[
    only marks,
    mark=*,
    mark size=1.2pt,
    color=orange,
    mark options={fill=orange}
  ] coordinates {
    (5, 5/39) (10, 5/17) (15, 105/253) (20, 95/189) (25, 25/44)
    (30, 435/703) (35, 85/129) (40, 65/94) (45, 495/689)
    (50, 1225/1653) (55, 165/217) (60, 885/1139) (65, 520/657)
    (70, 115/143) (75, 2775/3403) (80, 790/957) (85, 595/713)
    (90, 4005/4753) (95, 4465/5253) (100, 275/321) (105, 195/226)
    (110, 5995/6903) (115, 2185/2501) (120, 1785/2032)
    (125, 3875/4389) (130, 2795/3151) (135, 9045/10153)
    (140, 695/777) (145, 290/323) (150, 11175/12403)
    (155, 11935/13203) (160, 1060/1169) (165, 6765/7439)
    (170, 14365/15753) (175, 725/793) (180, 8055/8789)
    (185, 4255/4632) (190, 1995/2167) (195, 18915/20503)
    (200, 4975/5382)
  };

  \addplot[
    only marks,
    mark=*,
    mark size=1.2pt,
    color=green!60!black,
    mark options={fill=green!60!black}
  ] coordinates {
    (5, 1/19) (10, 3/20) (15, 7/29) (20, 38/119) (25, 5/13)
    (30, 29/66) (35, 17/35) (40, 52/99) (45, 33/59)
    (50, 245/416) (55, 99/161) (60, 118/185) (65, 52/79)
    (70, 23/34) (75, 185/267) (80, 632/893) (85, 119/165)
    (90, 267/364) (95, 893/1199) (100, 330/437) (105, 13/17)
    (110, 1199/1550) (115, 437/559) (120, 476/603)
    (125, 775/973) (130, 559/696) (135, 603/745) (140, 278/341)
    (145, 87/106) (150, 745/902) (155, 2387/2873) (160, 848/1015)
    (165, 451/537) (170, 2873/3404) (175, 145/171)
    (180, 1074/1261) (185, 851/995) (190, 1197/1394)
    (195, 1261/1463) (200, 3980/4601)
  };

\end{axis}
\end{tikzpicture}
	\end{subfigure}
	\begin{subfigure}{0.45\textwidth}
		\begin{tikzpicture}
\begin{axis}[
    width=\linewidth,
    height=0.7\linewidth,
    xmin=0, xmax=205,
    ymin=0.97, ymax=1.000,
    xlabel={$N$},
    ylabel={$F_e/p_{\mathrm{succ}}$},
    grid=both,
    legend style={at={(1.02,0.5)},anchor=west},
    legend cell align=left,
    tick label style={/pgf/number format/fixed},
]

  \addplot[
    only marks,
    mark=*,
    mark size=1.2pt,
    color=blue,
    mark options={fill=blue}
  ] coordinates {
    (2, 0.933013) (7, 0.98647) (12, 0.994039) (17, 0.996587)
    (22, 0.997766) (27, 0.998414) (32, 0.99881) (37, 0.999071)
    (42, 0.999253) (47, 0.999385) (52, 0.999484) (57, 0.99956)
    (62, 0.99962) (67, 0.999668) (72, 0.999708) (77, 0.99974)
    (82, 0.999768) (87, 0.999791) (92, 0.99981) (97, 0.999827)
    (102, 0.999842) (107, 0.999855) (112, 0.999866) (117, 0.999876)
    (122, 0.999885) (127, 0.999893) (132, 0.9999) (137, 0.999907)
    (142, 0.999913) (147, 0.999918) (152, 0.999923) (157, 0.999927)
    (162, 0.999931) (167, 0.999935) (172, 0.999938) (177, 0.999941)
    (182, 0.999944) (187, 0.999947) (192, 0.999949) (197, 0.999951)
    (202, 0.999954)
  };
  \addlegendentry{$d = 2$}

  \addplot[
    only marks,
    mark=*,
    mark size=1.2pt,
    color=orange,
    mark options={fill=orange}
  ] coordinates {
    (2, 0.971405) (12, 0.979059) (22, 0.9877) (32, 0.991892)
    (42, 0.994211) (52, 0.995633) (62, 0.996573) (72, 0.997229)
    (82, 0.997706) (92, 0.998066) (102, 0.998344) (112, 0.998564)
    (122, 0.998741) (132, 0.998886) (142, 0.999007) (152, 0.999108)
    (162, 0.999194) (172, 0.999267) (182, 0.999331) (192, 0.999386)
    (202, 0.999434)
  };
  \addlegendentry{$d = 3$}

  \addplot[
    only marks,
    mark=*,
    mark size=1.2pt,
    color=green!60!black,
    mark options={fill=green!60!black}
  ] coordinates {
    (2, 0.984123) (22, 0.982687) (42, 0.988192) (62, 0.991638)
    (82, 0.993765) (102, 0.995154) (122, 0.996112) (142, 0.996802)
    (162, 0.997316) (182, 0.997711) (202, 0.998021)
  };
  \addlegendentry{$d = 4$}

\end{axis}
\end{tikzpicture}
	\end{subfigure}
	\caption{Numerical values of $p_{\mathrm{succ}}$ of two consecutive successful pPBT teleportations employing a recycled resource, and the conditional entanglement fidelity $F_e(\mathcal{N})/p_{\mathrm{succ}}$ of the corresponding teleportation channel $\mathcal{N}$ for $d \in \set{2,3,4}$.}
	\label{fig:ent_fidelity_plot}
\end{figure}

\subsection{Comparison with deterministic multi-PBT scheme}
\label{ssec:multiPBT_comparison}
As has been said, two-step PBT falls within the regime of multi-port based teleportation schemes. Since the descrption of both optimal deterministic inexact ones as well as optimal probabilistic exact is known the performance of two-step pbt can be compared against them.

\begin{theorem}[Theorem~7 from~\cite{Mozrzymas2021optimalmultiport}]
	\label{thm:opt_mpbt}
	The entanglement fidelity of optimal multi-port based deterministic protocol is given by the maximal eigenvalue of so called teleportation matrix $M^{d,k}$
	\begin{equation}
		F_{\mathrm{ent}}= \frac{ 1}{ d^{2k}} \lambda_{\mathrm{max}}(M^{d,k}).
	\end{equation}

	The teleportation matrix is given by
	\begin{equation}
		M^{d,k}(N) = \left(R^d(N)\right)^T\dots\left(R^d(N-k+1)\right)^TR^d(N-k+1)\dots R^d(N)
	\end{equation}
	where $R^d(N) \in \mathbb{Z}^{p(n-1,d)\times p(n,d)}$. The matrix elements of $R^d(N)$ are defined as
	\begin{equation}
		[R^d(N)]_{\lambda\mu} = \begin{cases} 1 & \textrm{when } \exists
              \; a\in \AC_d(\lambda) : \lambda\cup a=\mu, \\
              0 & \textrm{otherwise.}\end{cases}
	\end{equation}
\end{theorem}

Observe that the entanglement fidelity in \cref{eq:F_e_answer} can also be used to get the performance of deterministic two-step recycling protocol.
In particular, the matrix $M$ depends only on the choice of measurements in a two-step PBT scheme and since no normalization of entanglement fidelity is required ($\mathcal N$ is trace-preserving, i.e. $p_\mathrm{succ}=1$), one can optimize the coefficients $\set{f_\mu}_\mu$ which define optimized resource state, so that they give the eigenvector corresponding to the maximal eigenvalue of $M$.
It gives the entanglement fidelity of the deterministic protocol, where the measurement effects $\Pi_0$ and $\widetilde{\Pi}_0$ are evenly distributed among the nonfailure effects $\Pi_i$ and $\widetilde{\Pi}_i$.

The matrix $M$ can be viewed as an analogue of the \emph{teleportation matrix} of the multi-dPBT teleportation from \cref{thm:opt_mpbt}.
Obtaining the exact formula for this eigenvalue as well as for the eigenvector seems beyond reach, so we provide numerical results in \cref{fig:ent_fidelity_mpbt_comparison} along with the performance of the corresponding multiport-based scheme.

\begin{figure}[H]
	\begin{tikzpicture}
  \begin{axis}[
    width=0.8\textwidth,
    height=0.4\textwidth,
    xmin=0, xmax=100,
    ymin=0.4, ymax=1,
    xlabel={$N$},
    ylabel={$F_e$},
    legend style={at={(1.02,0.5)},anchor=west},
    grid=both,
    tick label style={/pgf/number format/fixed},
  ]

    \addplot[
      only marks,
      mark=*,
      mark size=1.5pt,
      color = red,
    ] coordinates {
      (2,0.125)
      (4,0.475347)
      (6,0.682499)
      (8,0.792379)
      (10,0.854969)
      (12,0.893411)
      (14,0.918536)
      (16,0.935795)
      (18,0.948134)
      (20,0.95725)
      (22,0.964169)
      (24,0.969541)
      (26,0.973793)
      (28,0.977216)
      (30,0.980012)
      (32,0.982324)
      (34,0.984257)
      (36,0.985891)
      (38,0.987283)
      (40,0.988479)
      (42,0.989515)
      (44,0.990417)
      (46,0.991207)
      (48,0.991904)
      (50,0.992521)
      (52,0.99307)
      (54,0.993561)
      (56,0.994002)
      (58,0.994398)
      (60,0.994757)
      (62,0.995083)
      (64,0.995379)
      (66,0.995649)
      (68,0.995896)
      (70,0.996123)
      (72,0.996331)
      (74,0.996523)
      (76,0.9967)
      (78,0.996865)
      (80,0.997017)
      (82,0.997158)
      (84,0.99729)
      (86,0.997412)
      (88,0.997527)
      (90,0.997634)
      (92,0.997734)
      (94,0.997828)
      (96,0.997917)
      (98,0.998)
      (100,0.998078)
    };
    \addlegendentry{Multi PBT, $d=2 \quad$}

    \addplot[
      only marks,
      mark=square*,
      mark size=0.8pt,
      color = blue,
    ] coordinates {
      (2,0.116627)
      (4,0.462253)
      (6,0.672142)
      (8,0.784749)
      (10,0.84928)
      (12,0.889058)
      (14,0.915118)
      (16,0.933048)
      (18,0.945884)
      (20,0.955374)
      (22,0.962583)
      (24,0.968183)
      (26,0.972619)
      (28,0.97619)
      (30,0.979108)
      (32,0.981522)
      (34,0.983541)
      (36,0.985247)
      (38,0.986701)
      (40,0.987951)
      (42,0.989033)
      (44,0.989976)
      (46,0.990802)
      (48,0.99153)
      (50,0.992175)
      (52,0.99275)
      (54,0.993263)
      (56,0.993724)
      (58,0.994139)
      (60,0.994514)
      (62,0.994854)
      (64,0.995164)
      (66,0.995447)
      (68,0.995705)
      (70,0.995942)
      (72,0.99616)
      (74,0.996361)
      (76,0.996547)
      (78,0.996718)
      (80,0.996878)
      (82,0.997025)
      (84,0.997163)
      (86,0.997291)
      (88,0.997411)
      (90,0.997523)
      (92,0.997628)
      (94,0.997727)
      (96,0.997819)
      (98,0.997906)
      (100,0.997988)
    };
    \addlegendentry{Two-step PBT, $d=2$}

    \addplot[
      only marks,
      mark=triangle*,
      mark size = 2pt,
      color = red,
    ] coordinates {
      (5,0.218495)
      (10,0.540299)
      (15,0.7164)
      (20,0.811214)
      (25,0.86631)
      (30,0.900718)
      (35,0.923506)
      (40,0.939326)
      (45,0.950733)
      (50,0.959218)
      (55,0.965696)
      (60,0.970751)
      (65,0.974769)
      (70,0.978015)
      (75,0.980675)
      (80,0.98288)
      (85,0.98473)
      (90,0.986295)
      (95,0.987632)
      (100,0.988783)
    };
    \addlegendentry{Multi PBT, $d=3 \quad$}

    \addplot[
      only marks,
      mark=diamond*,
      mark size=1.5pt,
      color = blue, 
    ] coordinates {
      (5,0.21111)
      (10,0.530505)
      (15,0.708743)
      (20,0.805554)
      (25,0.862068)
      (30,0.897458)
      (35,0.920937)
      (40,0.937255)
      (45,0.949031)
      (50,0.957798)
      (55,0.964493)
      (60,0.96972)
      (65,0.973876)
      (70,0.977234)
      (75,0.979986)
      (80,0.982268)
      (85,0.984183)
      (90,0.985803)
      (95,0.987188)
      (100,0.988379)
    };
    \addlegendentry{Two-step PBT, $d=3$}

  \end{axis}
\end{tikzpicture}
	\caption{Comparison of deterministic two-step recycling PBT scheme with optimal deterministic multi PBT for two systems for $d \in \set{2,3}$.}
	\label{fig:ent_fidelity_mpbt_comparison}
\end{figure}

The low discrepancy between the optimal mPBT fidelity and two-step one is remarkable since in the latter case the fixed measurements are used (i.e. two optimal one-step pPBT PGMs) and the only optimization is done over the resource  (via coefficients $f_\mu$).


\section{Degradation of the resource state in probabilistic port-based teleportation}
\label{sec:degradation}
The result in the previous section decribes fully the quality of two-step protocol, however it is restricted to the repetitive application of PGM. However, in the probabilistic PBT with EPR resource, the POVM used is of different form~\cite{ishizaka_quantum_2009}. In order to study the possibility of reusing the resource one has to resort to indirect approach.
The quantity of interest now is how well the resource state is preserved after one (and possibly more) round of teleportation scheme. More specifically, we are interested in \emph{recycling fidelity}, i.e.~the overlap between the actual state after the measurement $\Psi_{\mathrm{out}}^{(i)}$ and the ideal state $\Psi_{\mathrm{id}}^{(i)}$ that corresponds to no resource degradation
\begin{equation}
	\label{eq:recycling_fidelity}
	F_{\mathrm{rec}}(N, d) = \sum_{ i=0}^{N} p_i F(\Psi_{\mathrm{out}}^{(i)}, \Psi_{\mathrm{id}}^{(i)}) = \sum_{ i=0}^{N} p_i \frac{ \braket{\Psi_{\mathrm{out}}^{(i)}}{\Psi_{\mathrm{id}}^{(i)}}}{ \norm{\ket{\Psi_{\mathrm{out}}^{(i)}}}_2 } ,
\end{equation}
where $i$ ranges over possible measurement outcomes and their corresponding post-measurement states. One can in particular discuss \emph{conditional recycling fidelity}
\begin{equation}
	\label{eq:cond_recycling_fidelity_succ}
	F_\mathrm{rec}^\mathrm{succ}(N, d) = \frac{ \sum_{ i=1}^{N} p_i F(\Psi_{\mathrm{out}}^{(i)}, \Psi_{\mathrm{id}}^{(i)})}{ \sum_{ i=1}^{N} p_i } = F(\Psi_{\mathrm{out}}^{(N)}, \Psi_{\mathrm{id}}^{(N)})
\end{equation}
and
\begin{equation}
	\label{eq:cond_recycling_fidelity_fail}
	F_\mathrm{rec}^\mathrm{fail}(N, d) = \frac{  p_0 F(\Psi_{\mathrm{out}}^{(0)}, \Psi_{\mathrm{id}}^{(0)})}{  p_0 } =F(\Psi_{\mathrm{out}}^{(0)}, \Psi_{\mathrm{id}}^{(0)}),
\end{equation}
where the last equality in~\eqref{eq:cond_recycling_fidelity_succ} follows from the covariance of the protocol with respect to the action of the symmetric group $\S_N$, which simply means that the teleportation characteristics cannot depend on the port where the teleported state arrives.
Thus, the quantities of interest are the overlaps of the post-measurement states and the corresponding ideal states.

\subsection{Degradation of the resource state in standard pPBT}
\label{ssec:degr_standard}
In the standard scheme, the parties share $N$ EPR pairs, so the input state is
\begin{equation}
	\label{eq:input_state_standard_pPBT}
	\ket{\Psi_\mathrm{in}} = \ket{\Psi^+}_{AB} \otimes \ket{\psi^+}_{\bar A_1R_1}.
\end{equation}
where $\ket{\psi^+}_{\bar A_1 R_1}$ is a maximally entangled state shared between Alice and a reference system, and the goal is to teleport the $\bar A_1$ system from Alice to Bob. The unnormalized outcome state after the measurement is
\begin{equation}
	\label{eq:post_measurement_a_pPBT}
	\ket{\Psi_{\mathrm{out}}^{(i)}} = \sqrt{\Pi_i}  \bigl(  \ket{\Psi^+}_{AB} \otimes \ket{\psi^+}_{\bar A_1R_1}\bigr).
\end{equation}
The ideal state corresponding to successful teleportation to the $i$-\emph{th} port ($i=1,\dots,N$) reads
\begin{equation}
	\label{eq:id_ith_state_pPBT}
	\ket{\Psi_{\mathrm{id}}^{(i)}} \defeq \bigotimes_{\substack{j = 1\\j \neq i}}^N	\ket{\psi^+}_{{A_j}{B_j}}\otimes\ket{\psi^+}_{{\bar A_1}{A_i}}\otimes \ket{\psi^+}_{{R_1}{B_i}}.
\end{equation}
Whenever the teleportation ends with failure, we postulate that the idealized state is the same as the initial one, since no teleportation takes place at all.
\begin{equation}
	\label{eq:id_0_state_pPBT}
	\ket{\Psi_{\mathrm{id}}^{(0)}} \defeq \ket{\Psi_\mathrm{in}} =  \ket{\Psi^+}_{AB} \otimes \ket{\psi^+}_{\bar A_1R_1}.
\end{equation}
Unrealistic this choice may seem, we are looking for the sufficient condition, so any state suitable for the teleportation of another state is worth examination.
%

\subsubsection{Conditional fidelity after success}
From \cite{Studziński_2022} we have the following relation:
\begin{equation}
	\of[\Big]{\of[\big]{(i,N) \sigma_N (i,N)} \otimes I_{B{R_1}}}
	\ket{\Psi_\mathrm{in}}=\ket{\Psi_{\mathrm{id}}^{(i)}}.
\end{equation}
For the maximally entangled resource state we get the following overlap
\begin{equation}
	\braket
	{\Psi_{\mathrm{out}}^{(i)}}
	{\Psi_{\mathrm{id}}^{(i)}}
	= \frac{1}{d^{N+1}} \Tr \left(\sqrt{\Pi_i} (i,N) \sigma_N (i,N) \right).
\end{equation}
The normalization is
\begin{align}
	\left|\left|(\sqrt{\Pi_i}_{\bar A_1 A}\otimes I_{R_1B})\ket{\Psi_\mathrm{in}}\right|\right|_2
	= \frac{1}{\sqrt{d^{N+1}}}
	\left(\Tr{\Pi_i}\right)^{1/2}.
\end{align}
Due to covariance, see \Cref{eq:cond_recycling_fidelity_succ}, it is enough to only consider $i=N$. The expression for recycling fidelity thus becomes
\begin{equation}
	\label{eq:F_succ_std_pPBT}
	F_\mathrm{rec}^\mathrm{succ}(N,d)=	\frac{ 1}{ \sqrt{d^{N+1}}}\frac{\Tr(\sqrt{\Pi_N}\sigma_N)}{\sqrt{\Tr{\Pi_N}}}.
\end{equation}

\begin{theorem}
	\label{thm:F_succ_std_pPBT}
	The conditional recycling fidelity in the case of success after one round of pPBT protocol is
	\begin{equation}
		\label{eq:F_succ_std_pPBT_explicit}
		F_{\mathrm{rec}}^{\mathrm{succ}} = \frac{1}{ \sqrt {d^{N-1}}}\frac{\sum_{\lambda\pt_d N-1} \sqrt{ \frac{ 1}{ \gamma(\lambda)} } m_\lambda d_\lambda}{\sqrt{\sum_{\lambda\pt_d N-1} \frac{ 1}{ \gamma(\lambda)} m_\lambda d_\lambda}}.
	\end{equation}
\end{theorem}
\begin{proof}

	The measurement in question is given by~\eqref{eq:standard_pPBT_measurement} and is a pseudo-projector and thus its square root is given by
	\begin{equation}
		\label{eq:pi_a_sq_std_pPBT}
		\sqrt{ \Pi_{N}} =	\sum_{\lambda\vdash_d N-1}\psi_{\bar A_1, A_N}^+\otimes \sqrt{ \frac{ d}{ \gamma(\lambda)} } P_{\lambda}^{{N^c}}.
	\end{equation}
	Thus, using $\sigma_N = \sqrt{d} \psi^+$, the expression for $\Tr (\sqrt{ \Pi_N} \sigma_N)$ becomes
	\begin{align}
		\Tr(\sqrt{\Pi_N}\sigma_N) & = d \sum_{\lambda\vdash_d N-1}\sqrt{\frac{ d}{ \gamma(\lambda)}} \Tr (\psi^+_{\bar A_1,A_N}\otimes P_\lambda) = d \sum_{\lambda\vdash_d N-1}\sqrt{\frac{ d}{ \gamma(\lambda)}} m_\lambda d_\lambda
	\end{align}
	and similarly
	\begin{equation}
		\Tr\Pi_N = \sum_{\lambda \vdash_d N-1} \frac{ d}{ \gamma(\lambda)} m_\lambda d_\lambda
	\end{equation}
	which substituted in~\eqref{eq:F_succ_std_pPBT} gives the desired result.
\end{proof}

\subsubsection{Conditional fidelity after failure}
The ideal state in the event of failure is the initial resource state, because one assumes that no teleportation takes place $\ket{\Psi^{(0)}_{\mathrm{id}}} = \ket{\Psi^+}_{AB}$. Thus, the overlap between actual post-failure resource state and the ideal state is
\begin{equation}
	F( \Psi^{(0)}_{\mathrm{id}}, \Psi^{(0)}_{\mathrm{out}})= \frac{ 1}{ \sqrt{ d^{N+1}} } \frac{ \Tr \sqrt{  \Pi_0}}{ \sqrt{ \Tr \Pi_0}  } .
\end{equation}
\begin{theorem}
	\label{thm:rec_fidelity_fail_std_pPBT}
	The conditional recycling fidelity in the case of failure in the standard pPBT protocol is
	\begin{equation}
		F^\mathrm{rec}_{\mathrm{fail}} = \frac{1}{\sqrt{d^{N+1}}}\frac{\Tr\sqrt{\Pi_0}}{\sqrt{\Tr\Pi_0}} =
		\frac{1 - \frac{1}{d^{N+1}} \sum_{\lambda \pt_{d} N-1} \sum_{a\in \AC_d{(\lambda)}} (1- \sqrt{1-
				g_{\lambda,a}}) m_\lambda d_{\lambda\cup{a}}}
		{\sqrt{ 1 - \frac{1}{d^{N+1}}\sum_{\lambda \pt_{d} N-1} \sum_{a\in \AC_d{(\lambda)}} g_{\lambda,a} m_\lambda d_{\lambda\cup{a}} }}.
	\end{equation}
	where
	\begin{equation}
		g_{\lambda, a} \defeq \frac{\gamma_{a}(\lambda) }{\gamma(\lambda)} = \frac{d+\cont(a)}{d + \lambda_1}
	\end{equation}
\end{theorem}
\begin{proof}
	The effect $\Pi_0$ is an element of the algebra of partially transposed permutation operators $\A_{N,1}^d$. It is given by $\Pi_0 = I - \hat{G}$  and in the Schur basis it is given by \cite{Grinko2024}
	\begin{equation}
		\label{eq:POVM_element_failure}
		\hat{G}^{(\lambda,\0)} = \sum_{\lambda \pt_d N-1}
		\sum_{a \in \AC_d(\lambda)}
		g_{\lambda,a} \Pi_{\lambda,a}, \qquad
		\Pi_{\lambda,a}
		= \sum_{S\in\Paths_N{(\lambda\cup{a})}}
		\proj{S \to (\lambda,\0)}.
	\end{equation}
	Since $\Tr \Pi_{\lambda,a} = m_\lambda d_{\lambda\cup{a}}$,
	\begin{equation}
		\label{eq:tr_pi_0}
		\Tr { \Pi_0}= d^{N+1} - \sum_{\lambda \pt_{d} N-1} \sum_{a\in \AC_d{(\lambda)}} g_{\lambda,a} m_\lambda d_{\lambda\cup{a}}
	\end{equation}
	as well as
	\begin{equation}
		\label{eq:tr_sqrt_pi_0}
		\Tr \sqrt{ \Pi_0}={d^{N+1} - \sum_{\lambda \pt_{d} N-1} \sum_{a\in \AC_d{(\lambda)}} (1- \sqrt{1-
			g_{\lambda,a}}) m_\lambda d_{\lambda\cup{a}}}.
	\end{equation}
	The second expression comes from the fact that $\Pi_{\lambda,a}$ is a projector so in every irrep $(\lambda,\0)$ and $ \sqrt{  \Pi_0}^{(\lambda,\0)}$ is a diagonal matrix with the entries given by $ \sqrt{  1-g_{\lambda,a}}$ and one has to substract the excessive 1's from those irreps from the full trace of $I$.
	Thus \cref{eq:tr_pi_0,eq:tr_sqrt_pi_0} lead to the statement of the Theorem.
\end{proof}
%

\subsubsection{Summary for the standard pPBT}
The numerical values of the conditional recycling fidelity both in case of success and failure are depicted in \cref{fig:recycling_fidelity_std_pPBT}.

\begin{figure}[ht!]
	\centering
	\begin{subfigure}{.45\textwidth}
		\begin{tikzpicture}
\begin{axis}[
    width=\linewidth,
    height=0.7\linewidth,
    xmin=0, xmax=100,
    ymin=0.9979, ymax=0.9997,
    xlabel={$N$},
    ylabel={$F_{\mathrm{succ}}^{\mathrm{rec}}$},
    grid=both,
    tick label style={/pgf/number format/fixed},
    legend style={at={(1.02,0.5)},anchor=west},
    legend cell align=left,
    scaled y ticks = false, 
    yticklabel style={
        /pgf/number format/fixed,
        /pgf/number format/precision=4
    },
]

  \addplot[
    only marks,
    mark=*,
    mark size=1.5pt,
    color=blue,
    mark options={fill=blue}
  ] coordinates {
    (5,0.997977) (10,0.998375) (15,0.998591) (20,0.998798)
    (25,0.998938) (30,0.999056) (35,0.999145) (40,0.999221)
    (45,0.999282) (50,0.999335) (55,0.999379) (60,0.999419)
    (65,0.999452) (70,0.999483) (75,0.99951)  (80,0.999534)
    (85,0.999556) (90,0.999576) (95,0.999593) (100,0.99961)
  };

  \addplot[
    only marks,
    mark=*,
    mark size=1.5pt,
    color=orange,
    mark options={fill=orange}
  ] coordinates {
    (5,0.998409) (10,0.998345) (15,0.998466) (20,0.998597)
    (25,0.998715) (30,0.998815) (35,0.998901) (40,0.998975)
    (45,0.999039) (50,0.999096) (55,0.999146) (60,0.99919)
    (65,0.999229) (70,0.999265) (75,0.999298) (80,0.999327)
    (85,0.999354) (90,0.999379) (95,0.999401) (100,0.999423)
  };

\end{axis}
\end{tikzpicture}
	\end{subfigure}
	\hspace{1mm}
	\begin{subfigure}{.45\textwidth}
		\begin{tikzpicture}
\begin{axis}[
    width=\linewidth,
    height=0.7\linewidth,
    xmin=0, xmax=100,
    ymin=0.70, ymax=0.92,
    xlabel={$N$},
    ylabel={$F_{\mathrm{fail}}^{\mathrm{rec}}$},
    grid=both,
    tick label style={/pgf/number format/fixed},
    legend style={at={(1.02,0.5)},anchor=west},
    legend cell align=left,
]

  \addplot[
    only marks,
    mark=*,
    mark size=1.5pt,
    color=blue,
    mark options={fill=blue}
  ] coordinates {
    (5,0.806503) (10,0.782474) (15,0.766308) (20,0.758314)
    (25,0.751239) (30,0.746897) (35,0.742707) (40,0.739874)
    (45,0.737024) (50,0.734985) (55,0.732886) (60,0.731326)
    (65,0.729696) (70,0.728452) (75,0.72714)  (80,0.726117)
    (85,0.72503)  (90,0.724169) (95,0.72325)  (100,0.722512)
  };
  \addlegendentry{$d = 2$}

  \addplot[
    only marks,
    mark=*,
    mark size=1.5pt,
    color=orange,
    mark options={fill=orange}
  ] coordinates {
    (5,0.90547)  (10,0.884583) (15,0.872636) (20,0.86466)
    (25,0.858814) (30,0.854273) (35,0.850605) (40,0.847558)
    (45,0.844971) (50,0.842736) (55,0.84078)  (60,0.839047)
    (65,0.837497) (70,0.836101) (75,0.834833) (80,0.833675)
    (85,0.832611) (90,0.83163)  (95,0.83072)  (100,0.829874)
  };
  \addlegendentry{$d = 3$}

\end{axis}
\end{tikzpicture}
	\end{subfigure}
	\caption{Conditional recycling fidelity for success (left) and failure (right) in standard pPBT scheme, for $d=2$ (blue dots) and $d=3$ (orange dots).\label{fig:recycling_fidelity_std_pPBT}}
\end{figure}

Provided that the first round of the teleportation succeeds, allowing for arbitrarily large resources, the resulting resource state converges to the ideal state. In that case, the remaining ports can be reused for another round of pPBT. In case of failure, the problem is unresolved -- one cannot infer whether the remaining resource is sufficient for a further application of the standard pPBT protocol. However, it is known~\cite{ishizaka_quantum_2009, Studzinski2017, majenz2} that $p_{\mathrm{fail}}=1-p_{\mathrm{succ}} \to 0$ when $N\to\infty$. This means that $\lim_{N\to \infty}F_{\mathrm{rec}}(N, d)= \lim_{N\to \infty}F_{\mathrm{succ}}^{\mathrm{rec}}(N, d)$, meaning that an arbitrarily large resource state can be reused after performing one round of the standard pPBT protocol.

\subsection{Degradation of the resource state for pPBT with optimised resource}
Although the numerical analysis of the \Cref{thm:two_step_fidelity,thm:two_step_psucc} show that two step pPBT with optimised resource is a faithful protocol for large $N$, for the completeness of the paper the analysis of the degradation of the resource state after the first round of the protocol is presented in the following subsections.
\subsubsection{Conditional fidelity after success}

If the teleportation procedure succeeds, the ideal state
\begin{equation}
	\label{eq:id_state_a_opt_pPBT}
	\ket{\Psi_{\mathrm{id}}^{(i)}} =
	\ket{\psi^+}_{{\bar A_1}{A_i}} \otimes
	\ket{\psi^+}_{{R_1}{B_i}} \otimes
	\bigl(O_{\widetilde A}\otimes I_{}\bigr) \bigotimes_{j\neq i}
	\ket{\psi^+}_{{A_j}{B_j}}
\end{equation}
can be written as
\begin{equation}
	\label{eq:id_state_a_opt_pPBT_sigma}
	\ket{\Psi_{\mathrm{id}}^{(i)}}
	= \sof[\big]{(i,N) \sigma_N (i,N)}
	\ket{\psi^+}_{\bar A_1R_1}
	\ket{\psi^+}_{A_iB_i} \otimes
	\bigl(O_{\widetilde A} \otimes I\bigr)
	\bigotimes_{j\neq i} \ket{\psi^+}_{{A_j}{B_j}},
\end{equation}
whereas the actual unnormalized outcome state is
\begin{equation}
	\label{eq:post_measurement_a_opt_pPBT}
	\ket{\Psi_{\mathrm{out}}^{(i)}} = \sqrt{\Pi_i}  \bigl(  (O_A \otimes I_{B} )\ket{\Psi^+}_{AB} \otimes \ket{\psi^+}_{\bar A_1R_1}\bigr).
\end{equation}
where $O_{\widetilde{A}}$ is an optimization operator that acts on $\widetilde{A}=1,\dots, N-1$, that is on the new set of indices on Alice's side, after performing the first teleportation (under the assumption that the first state was teleported to $N$th port).
Due to the covariance of the protocol, see~\eqref{eq:cond_recycling_fidelity_succ}, one can restrict oneself to the study of one effect, namely $\Pi_N$. The overlap between the ideal and the unnormalized outcome resource state (for $i=N$) is then given by
\begin{align}
	\label{eq:F_succ_opt_pPBT_unnormalised}
	F_\mathrm{rec}^\mathrm{succ} = F(\Psi_{\mathrm{out}}^{(N)}, \Psi_{\mathrm{id}}^{(N)}) = \frac{ 1}{ \sqrt{ d^{N+1}} } \frac{ \Tr (\sigma_N \sqrt{ \Pi_N} O_AO_{\widetilde{A}})}{\sqrt{\Tr(\Pi_N O_A O_A)}}.
\end{align}

\begin{theorem}[Theorem 14 from~\cite{Studziński_2022}, rephrased]
	The recycling fidelity after success in one round of optimal pPBT is given by
	\begin{equation}
		F_\mathrm{rec}^\mathrm{succ} = \frac{ 1}{ d^{ \frac{ N+1}{2 } }} \Bigl(  \sum_{\substack{ \lambda \vdash_d N-1}} \sum_{a \in \AC_d(\lambda)} c_\lambda^{1/2} c_{\lambda\cup{a}}^{1/2}  \frac{ d_\lambda \sum_{ a'\in \AC_d(\lambda)}^{} \sqrt{ m_{\lambda\cup{a'}} d_{\lambda\cup{a'}}}} { \sqrt{ \sum_{ a\in \AC_d(\lambda)}^{} d_{\lambda\cup {a}}} }  \frac{ \sqrt{ m_{\lambda\cup{a}} d_{\lambda\cup{a}}} }{ \sqrt{  N d_\lambda }}  \Bigr) \Bigl/ \sqrt{\sum_{ \lambda \vdash_d N-1}^{} u_\lambda d_\lambda m_\lambda},
	\end{equation}
	where we defined
	\begin{equation}
		u_\lambda \defeq \frac{ d^{N+1} g(N) m_\lambda}{N d_\lambda } = c_\lambda \frac{d^2 \, g(N)}{N \, g(N-1)}.
	\end{equation}
\end{theorem}

\begin{proof}
	Observe that \cref{eq:F_succ_opt_pPBT_unnormalised} is of the same form as the expression for the overlap in optimal dPBT scheme \cite{Studziński_2022}.
	The measurement in the optimal pPBT scheme is closely related to POVMs in standard and optimized dPBT schemes, up to the additional factor $\Pi_0/N$, which does not contribute to the quantity since it is orthogonal to $\sigma_N$. The quantity is thus the same, up to different coefficients of the optimization procedure $O_A,O_{\widetilde A}$, which were substituted in the final expression.
\end{proof}

\subsubsection{Conditional fidelity after failure}

In case of optimal protocol, the ideal state is
\begin{equation}
	\label{eq:id_0_state_opt_pPBT}
	\ket{\Psi_{\mathrm{id}}} = (O_A\otimes I_{B} )
	\ket{\Psi^+}_{AB}
	\ket{\psi^+}_{\bar A_1R_1}
\end{equation}
and the (unnormalised) post-measurement state is
\begin{equation}
	\label{eq:post_measurement_0_opt_pPBT}
	\ket{\Psi_{\mathrm{out}}^{(0)}} =
	\sqrt{\Pi_0}
	\bigl(
	(O_A \otimes I_{B})
	\ket{\Psi^+}_{AB} \otimes
	\ket{\psi^+}_{\bar A_1R_1}
	\bigr).
\end{equation}
Since \cite{Grinko2024} the effect $\Pi_0$ is identity in irreps $(\mu, \square)$ and zero in other ones, one has the following expression for the overlap
\begin{align}
	\Tr( \ket{\Psi_{\mathrm{id}}}\bra{\Psi_{\mathrm{out}}}) = & \Tr( O_A\ket{\psi^+_{AB}}\ket{\psi^+_{\bar A_1R_1}} \bra{\psi^+}_{AB}\bra{\psi^{+}}_{\bar A_1R_1}O_A^\dagger\sqrt{\Pi_0} )                                                            \\
	=                                                         & \frac{ 1}{ d^{N+1}} \Tr (O_AO_A^\dagger \sqrt{\Pi_0}) 			=                                        \sum_{ \mu \vdash_{d-1} N}^{} \frac{ c_\mu}{ d^{N+1}} \Tr (P_\mu M_{(\mu,\square)})
\end{align}
where $M_{(\mu,\square)}$ is a projection onto irrep $(\mu,\square)$. Since $\Tr (P_\mu M_{(\mu,\square)}) = m_{(\mu,\square)} d_\mu$, and the normalisation factor is

\begin{equation}
	\begin{split}
		||\ket{\Psi_{\mathrm{out}}}|| = \frac{ 1}{ \sqrt { d^{N+1}} } \Bigl( \sum_{ \mu\vdash_{d-1} N } c_\mu\Tr (P_\mu M_{(\mu, \square)})\Bigr) ^{1/2}
	\end{split}
\end{equation}
one obtains fhe following
\begin{theorem}
	\label{thm:rec_fidelity_fail_opt_pPBT}
	The recycling fidelity in case of failure after one round of optimal pPBT protocol is given by
	\begin{equation}
		\label{eq:rec_fidelity_fail_opt_pPBT}
		F^\mathrm{rec}_\mathrm{fail}= \frac{ 1}{  \sqrt{ d^{N+1}} } \Bigl( \sum_{ \mu\vdash_{d-1} N }  c_\mu m_{(\mu,\square)} d_\mu\Bigr)^{1/2}
	\end{equation}
\end{theorem}

\subsubsection{Summary of optimal pPBT}

Our numerical results regarding conditional recycling fidelity both in the case of success and failure are depicted in \cref{fig:recycling_fidelity_opt_pPBT}.

\begin{figure}[!ht]
	\centering
	\begin{subfigure}{.45\textwidth}
		\begin{tikzpicture}
\begin{axis}[
    width=\linewidth,
    height=0.7\linewidth,
    xmin=0, xmax=100,
    ymin=0.83, ymax=0.96,
    xlabel={$N$},
    ylabel={$F_{\mathrm{succ}}^{\mathrm{rec}}$},
    grid=both,
    tick label style={/pgf/number format/fixed},
    yticklabel style={
        /pgf/number format/fixed,
        /pgf/number format/fixed zerofill,
        /pgf/number format/precision=2
    },
]

  \addplot[
    only marks,
    mark=*,
    mark size=1.5pt,
    color=blue,
    mark options={fill=blue}
  ] coordinates {
    (5,0.930696) (10,0.913664) (15,0.906837) (20,0.903080)
    (25,0.900677) (30,0.899000) (35,0.897757) (40,0.896797)
    (45,0.896032) (50,0.895407) (55,0.894887) (60,0.894446)
    (65,0.894068) (70,0.893740) (75,0.893453) (80,0.893198)
    (85,0.892972) (90,0.892769) (95,0.892586) (100,0.892420)
  };

  \addplot[
    only marks,
    mark=*,
    mark size=1.5pt,
    color=orange,
    mark options={fill=orange}
  ] coordinates {
    (5,0.949663) (10,0.913292) (15,0.892687) (20,0.879627)
    (25,0.870633) (30,0.864059) (35,0.859038) (40,0.855073)
    (45,0.851859) (50,0.849198) (55,0.846956) (60,0.845040)
    (65,0.843382) (70,0.841933) (75,0.840655) (80,0.839519)
    (85,0.838502) (90,0.837585) (95,0.836755) (100,0.835999)
  };

\end{axis}
\end{tikzpicture}
	\end{subfigure}
	\hspace{1mm}
	\begin{subfigure}{.45\textwidth}
		\begin{tikzpicture}
\begin{axis}[
    width=\linewidth,
    height=0.7\linewidth,
    xmin=0, xmax=100,
    ymin=0.16, ymax=0.8,
    xlabel={$N$},
    ylabel={$F_{\mathrm{fail}}^{\mathrm{rec}}$},
    grid=both,
    tick label style={/pgf/number format/fixed},
    yticklabel style={
        /pgf/number format/fixed,
        /pgf/number format/fixed zerofill,
        /pgf/number format/precision=1
    },
    legend style={at={(1.02,0.5)},anchor=west},
    legend cell align=left,
]

  \addplot[
    only marks,
    mark=*,
    mark size=1.5pt,
    color=blue,
    mark options={fill=blue}
  ] coordinates {
    (5,0.612372) (10,0.480384) (15,0.408248) (20,0.361158)
    (25,0.327327) (30,0.301511) (35,0.280976) (40,0.264135)
    (45,0.250000) (50,0.237915) (55,0.227429) (60,0.218218)
    (65,0.210042) (70,0.202721) (75,0.196116) (80,0.190117)
    (85,0.184637) (90,0.179605) (95,0.174964) (100,0.170664)
  };
  \addlegendentry{$d = 2$}

  \addplot[
    only marks,
    mark=*,
    mark size=1.5pt,
    color=orange,
    mark options={fill=orange}
  ] coordinates {
    (5,0.784465) (10,0.666667) (15,0.589768) (20,0.534522)
    (25,0.492366) (30,0.458831) (35,0.431331) (40,0.408248)
    (45,0.388514) (50,0.371391) (55,0.356348) (60,0.342997)
    (65,0.331042) (70,0.320256) (75,0.310460) (80,0.301511)
    (85,0.293294) (90,0.285714) (95,0.278693) (100,0.272166)
  };
  \addlegendentry{$d = 3$}

\end{axis}
\end{tikzpicture}
	\end{subfigure}
	\caption{Conditional recycling fidelity for success (left) and failure (right) in optimal pPBT scheme, for $d=2$ (blue dots) and $d=3$ (orange dots).}
	\label{fig:recycling_fidelity_opt_pPBT}
\end{figure}

One has to note that the results presened in \Cref{fig:recycling_fidelity_opt_pPBT} are not in disagreement with the \Cref{thm:two_step_fidelity}, since recycling fidelity $F_\mathrm{rec}$ gives only a sufficient condition.


\section{Conclusions and open questions}

Two main results are presented in this work. Firstly, a detailed analysis of the case of two subsequent successful applications of PBT protocols employing PGM is compared to the teleportation of two quantum states at once, via a multiport-based scheme. This result shows that although such an approach provides smaller entanglement fidelity, the numerics show that for large $N$ the difference vanishes. Thus not only can two states be teleported with the repetitive use of a simpler measurement than a bigger joint measurement employed in the MPBT scheme, but it also means that two teleportations can be performed with a time delay.

What is striking, the two-step deterministic teleportation is close in terms of entanglement fidelity to the MPBT scheme even for small $N$, suggesting that  optimising the second measurement together with keeping the whole resource after the first round, could match the performance of multiport scheme for mutual teleportation of two quantum states.

Moreover, the analysis of resource degradation is performed in the pPBT scheme. It is shown that in case of EPR resource, the degradation vanishes as the number of ports goes to infinity. This means that such a resource can be reused for further application of the pPBT protocol, meaning that it is an economical way to use a quantum resource, the preparation of which can be costly.

An open problem that arises from this work is developing an adapted measurement suited to the new distorted resource. This requires a reformulation of the problem in the form of a semidefinite program and is addressed by the ongoing research.

\section*{Acknowledgements}

MS is supported by the National Science Centre, Poland, Grant Sonata 16 no. 2020/39/D/ST2/01234.
PK is supported by the National Science Centre, Poland, Grant Preludium 20, no. 2021/41/N/ST2/03249.
MO and DG are supported by an NWO Vidi grant (No.VI.Vidi.192.109) and by a National Growth Fund grant (NGF.1623.23.025).
AB is supported by an NWO Vidi grant (Project No VI.Vidi.192.109) and by the European Union (ERC, ASC-Q,
101040624). Views and opinions expressed are however those of the authors
only and do not necessarily reflect those of the European Union or the
European Research Council. Neither the European Union nor the granting
authority can be held responsible for them.

\appendix
\section{Proof of \texorpdfstring{\cref{thm:two_step_fidelity}}{}}
\label{apx:proof}

By the cyclic property of trace,
\begin{equation}
	\label{eq:fidelity_2}
	F_e(\mathcal{N}) = \frac{N(N-1)}{d^{N+4}} \Tr \Big(\sqrt{\Pi_N} \; \widetilde{\Pi}_{N-1} \; \sqrt{\Pi_N} \; O \; \tau_2 \; O  \Big).
\end{equation}
As all POVM elements are unitary-equivariant, we may express all elements of \eqref{eq:fidelity_formula} in the Gelfand--Tsetlin basis. Therefore \eqref{eq:fidelity_2} reads
\begin{equation}
	\label{eq:fidelity_3}
	F_e(\mathcal{N})= \frac{N(N-1)}{d^{N+4}}  \sum_{\Lambda\in \Irr{\A_{N,2}^{d}}} m_\Lambda \Tr
	\sof*{\sqrt{\Pi_N^\Lambda} \; \widetilde{\Pi}_{N-1} \; \sqrt{\Pi_N^\Lambda} \; O^\Lambda \;  \tau_2^\Lambda \; O^\Lambda  }
\end{equation}
where $m_\Lambda$ is the dimension of the corresponding unitary irrep, and $\tau_2^{\Lambda}$ is an irrep $\Lambda$ of $\tau_2$ from \cref{eq:tau_2}.
Notice that, since $\Pi_N^\Lambda, \; \widetilde{\Pi}_{N-1}^\Lambda $ are supported only on irreps $ \Lambda = (\nu,\0)$, the only contributions to \eqref{eq:fidelity_3} are from those irreps. Therefore \eqref{eq:fidelity_3} is equivalent to
\begin{align}
	\label{eq:fidelity_4}
	F_e(\mathcal{N}) & = \frac{N(N-1)}{d^{N+4}}  \sum_{\nu\pt_d N-2} m_\nu
	\Tr \sof*{
	\widetilde{\Pi}_{N-1}^{(\nu,\0)} \; \sqrt{\Pi_N^{(\nu,\0)}} \;O^{(\nu,\0)}\;  \tau_2^{(\nu,\0)} \; O^{(\nu,\0)} \; \sqrt{\Pi_N^{(\nu,\0)}}
	}.
\end{align}
We shall derive formulas for each element in \eqref{eq:fidelity_4}.

Firstly, notice that
\begin{equation}
	\label{eq:sub0}
	( \sqrt{\Pi_N})^{(\nu,\0)} =\sum_{a\in \AC_d (\nu )} \sum_{r\in \RC (\nu\cup a )}  \sum_{S\in \Paths (\nu\cup a /r)} \frac{\ket{w^\nu_{S,r,a}}\bra{{w^\nu_{S,r,a}}}}{|| \ket{w^\nu_{S,r,a} }||}
\end{equation}
where
\begin{equation}
	\label{eq:sub1}
	\ket{w^\nu_{S,r,a}} =\sum_{b\in \AC_d (\nu \cup a)} \sqrt{\frac{d_{\nu\cup a\cup b}}{N d_{\nu\cup a}}} \ket{S \to \nu\cup a \to \nu \cup a \cup b \to (\nu \cup a,\0) \to (\nu,\0)}.
\end{equation}
Second, we can express $\widetilde{\Pi}_{N-1}$ using \eqref{eq:PI_N}, namely
\begin{equation}
	\label{eq:sub2}
	\widetilde{\Pi}_{N-1}^{\Lambda} =
	\sum_{\nu\pt_d N-2}
	\sum_{S\in \Paths (\nu )}
	\widetilde{G}_{\nu,S}^\Lambda \,(U^\Lambda\,\sigma_N^\Lambda \,U^\Lambda )\, \widetilde{G}_{\nu,S}^\Lambda
\end{equation}
where $U^\Lambda=\sigma_{N-1}^\Lambda\sigma_{N+1}^\Lambda$ and hence $U=U^\dagger$ and $U^\Lambda\sigma_N^\Lambda U^\Lambda$ is a contraction between $N-1$ and $N+2$, and
\begin{equation}
	\label{eq:sub3}
	\widetilde{G}_{\nu,S}
	= \sum_{a\in \AC_d (\nu )}
	\frac{1}{\sqrt{d+\cont{(a)}}} E_{S \to \nu \cup a,S \to \nu \cup a}
	\otimes I^{\otimes 3}.
\end{equation}
Notice that $\sigma_N^\Lambda$ for $\Lambda = (\nu,\0)$ has the following form:
\begin{equation}
	\label{eq:sub4}
	\sigma_N^{(\nu,\0)}
	=
	\sum_{a\in \AC_d (\nu)} \sum_{r\in \RC (\nu\cup a)} \sum_{S\in \Paths (\nu\cup a/ r)}
	\frac{\ket{\widetilde{v^\nu_{S,r,a}}}\bra{\widetilde{v^\nu_{S,r,a}}}}{|| \ket{\widetilde{v^\nu_{S,r,a}}}||}
\end{equation}
where
\begin{equation}
	\label{eq:sub5}
	\ket{\widetilde{v^\nu_{S,r,a}}} =
	\sum_{b\in \AC_d(\nu \cup a)}
	\sqrt{\frac{m_{\nu\cup a\cup b} }{m_{\nu\cup a}}}
	\ket{S \to \nu\cup a \to \nu \cup a \cup b \to (\nu \cup a,\0) \to (\nu,\0)}
\end{equation}
Third, we compute
\begin{equation}
	\label{eq:sub6}
	O^{(\nu,\0)}\;  \tau_2^{(\nu,\0)} \; O^{(\nu,\0)} =  \sum_{S\in \Paths_{N-2}(\nu)}
	\ket{v^\nu_{S}} \bra{v^\nu_{S}}
\end{equation}
where
\begin{equation}
	\label{eq:sub7}
	\ket{v^\nu_{S}}=\sum_{a\neq b\in \AC_d (\nu )} \sqrt{c_{\nu \cup a \cup b}} \sqrt{\frac{m_{\nu\cup a\cup b}}{m_{\nu}}} \ket{S \to \nu\cup a \to \nu \cup a \cup b \to (\nu \cup a,\0) \to (\nu,\0)}.
\end{equation}
Lastly, we denote by
\begin{equation}
	\label{eq:sub8}
	\ket{\widetilde{w_{S,a,a}^\nu}} \defeq G_{\nu,S} \ket{w_{S,a,a}^\nu}=
	\sum_{b\in \AC_d (\nu \cup a)} \sqrt{\frac{d_{\nu\cup a\cup b}}{N d_{\nu\cup a} (d+\cont(a))}} \ket{S \to \nu\cup a \to \nu \cup a \cup b \to (\nu \cup a,\0) \to (\nu,\0)}.
\end{equation}
Combining \cref{eq:sub0,eq:sub1,eq:sub2,eq:sub3,eq:sub4,eq:sub5,eq:sub6,eq:sub7,eq:sub8}, the entanglement fidelity \eqref{eq:fidelity_4} becomes
\begin{multline}
	\label{eq:fidelity_5}
	F_e(\mathcal{N})= \frac{N(N-1)}{d^{N+4}}
	\sum_{\nu\pt_d N-2}
	d_\nu m_\nu \\ \times
	\sum_{\substack{a, a', a'' \\ \in \AC_d (\nu)}}
	\sum_{\substack{r \in \RC(\nu \cup a) \\ r' \in \RC(\nu \cup a') \\ r'' \in \RC(\nu \cup a'')}}
	\frac{\braket{v_S^\nu }{w_{S,a,r}^\nu}}{||w_{S,a,r}^\nu||}
	\braket{\widetilde{w_{S,a,r}^\nu}}{U^{(\nu,\0)} | \widetilde{v_{S,a',r'}^\nu}}
	\braket{ \widetilde{v_{S,a',r'}^\nu}}{U^{(\nu,\0)}|\widetilde{w_{S,a'',r''}^\nu}}
	\frac{\braket{w_{S,a'',r''}^\nu}{v_S^\nu }}{||w_{S,a'',r''}^\nu||}
\end{multline}
Notice that $\braket{w_{S,a,r}^\nu}{v_S^\nu }=0$ for $a \neq r$, hence \eqref{eq:fidelity_5} becomes
\begin{equation}
	\label{eq:fidelity_6}
	F_e(\mathcal{N})= \frac{N(N-1)}{d^{N+4}}
	\sum_{\nu\pt_d N-2}
	d_\nu m_\nu
	\sum_{\substack{a,a',a'' \\ \in \AC_d (\nu )}}
	\frac{\braket{v_S^\nu }{w_{S,a,a}^\nu}}{||w_{S,a,a}^\nu||}
	\braket{\widetilde{w_{S,a,a}^\nu}}{U^{(\nu,\0)} | \widetilde{v_{S,a',a'}^\nu}}
	\braket{ \widetilde{v_{S,a',a'}^\nu}}{U^{(\nu,\0)}|\widetilde{w_{S,a'',a''}^\nu}}
	\frac{\braket{w_{S,a'',a''}^\nu}{v_S^\nu }}{||w_{S,a'',a''}^\nu||}.
\end{equation}
Now we compute the overlaps in \cref{eq:fidelity_6}. First, notice from \cref{eq:sub1,eq:sub3} that
\begin{equation}
	\label{eq:overlap_1}
	\braket{v_S^\nu }{w_{S,a,a}^\nu}=
	\sum_{b\in \AC_d (\nu \cup a)}
	\sqrt{\frac{d_{\nu\cup a\cup b}}{N d_{\nu\cup a}}} \sqrt{c_{\nu \cup a \cup b}} \sqrt{\frac{m_{\nu\cup a\cup b}}{m_{\nu}}}.
\end{equation}
And further using the definition of $c_{\nu\cup a \cup b}$, see \cref{eq:c_mu_pPBT}, we arrive at
\begin{equation}
	\label{eq:overlap_1b}
	\braket{v_S^\nu }{w_{S,a,a}^\nu} = \sqrt{\frac{d^N}{N d_{\nu\cup a}m_\nu}}
	\of[\Big]{\sum_{b\in \AC_d (\nu \cup a)} \sqrt{f_{\nu \cup a\cup b}}}.
\end{equation}
For optimal pPBT state we have, using \cref{eq:c_mu_pPBT}, we have
\begin{equation}
	\braket{v_S^\nu }{w_{S,a,a}^\nu}
	=  \of[\Big]{\sum_{b\in \AC_d (\nu \cup a)} m_{\nu \cup a\cup b}} \sqrt{\frac{d^N g(N)}{N d_{\nu\cup a}m_\nu}}
	=  d\, m_{\nu \cup a} \sqrt{\frac{d^N g(N)}{N d_{\nu\cup a}m_\nu}}.
\end{equation}

Secondly, we have that
\begin{equation}
	\label{eq:overlap_2}
	\braket{\widetilde{w_{S,a,a}^\nu}}{U^{(\nu,\0)} | \widetilde{v_{S,a',a'}^\nu}}
	=\sum_{\substack{b\in \AC (\nu\cup a )\\b'\in \AC (\nu \cup a')}}
	\frac{1}{\sqrt{d+\cont{(a)}}}
	\sqrt{\frac{d_{\nu\cup a\cup b}m_{\nu\cup a'\cup b'}}{d_{\nu\cup a}m_{\nu\cup a'}}}
	\bra{\substack{S \to \nu\cup a \to \nu \cup a \cup b \to \\ \to (\nu \cup a,\0) \to (\nu,\0)}}
	U
	\ket{\substack{S \to \nu\cup a' \to \nu \cup a' \cup b' \to \\ \to (\nu \cup a',\0) \to (\nu,\0)}}.
\end{equation}
Notice that non-vanishing contributions to such expressions are if and only if $\{a,b\}=\{a',b'\}$, hence we can distinguish two cases: either $a=a'$ and $b=b'$, or $b=a'\neq a=b'$.
Utilizing Young--Yamanouchi basis formulas for $U=\sigma_{N-1}^\Lambda\sigma_{N+1}^\Lambda$, we arrive at:
\begin{multline}
	\label{eq:overlap_2b}
	\braket{\widetilde{w_{S,a,a}^\nu}}{U^{(\nu,\0)} | \widetilde{v_{S,a',a'}^\nu}}
	= \\[5pt] =
	\begin{cases}
		\displaystyle \sum_{b\in \AC_d(\nu \cup a)}  \frac{m_{\nu\cup a\cup b}}{m_{\nu\cup a}}
		\frac{1}{\sqrt{(d+\cont{(a)})(d+\cont{(b)})}} \frac{1}{(\cont{(a)-\cont{(b)}})^2}                                                                                             &
		\text{if $a=a'\in \AC_d(\nu)$},                                                                                                                                                 \\
		\displaystyle \frac{m_{\nu\cup a\cup a'}}{\sqrt{m_{\nu\cup a}m_{\nu\cup a'}}} \frac{1}{\sqrt{(d+\cont{(a)})(d+\cont{(a')})}} \Big(1-\frac{1}{(\cont{(a)-\cont{(a')}})^2}\Big) &
		\text{if $a\neq a'\in \AC_d(\nu)$}.
	\end{cases}
\end{multline}
After rewriting of \cref{eq:fidelity_5}, we get the expression:
\begin{align}
	F_e(\mathcal{N}) & = \frac{1}{d^4} \sum_{\nu\pt_d N-2} \sum_{\substack{a,a',a'' \in \AC_d (\nu )}} \sum_{\mu,\mu'' \pt_d N} v_\mu S^{\nu}_{a,\mu} H^\nu_{a,a'} H^\nu_{a',a''} S^{\nu}_{a'',\mu''} v_{\mu''} = \frac{1}{d^4} v\tp M v,
\end{align}
where the vector $v$ has components labeled every partition $\mu \pt_d N$:
\begin{align}
	v_\mu & \defeq \sqrt{f_\mu},
\end{align}
and the matrices $H^\nu, S^{\nu}$ and $M$ are defined as
\begin{align}
	H^\nu_{a,a'}    & \defeq
	\begin{dcases}
		\sum_{b\in \AC_d(\nu \cup a)} \sqrt{\frac{d + \cont(b)}{d + \cont(a)}} \frac{q(b|\nu \cup a)}{(\cont(a)-\cont(b))^2} &
		\text{if $a=a'\in \AC_d(\nu)$},                                                                                        \\
		\sqrt{q(a|\nu \cup a')q(a'|\nu \cup a)} \of*{1- \frac{1}{(\cont(a)-\cont(a'))^2}}                                    &
		\text{if $a\neq a'\in \AC_d(\nu)$},
	\end{dcases}                              \\
	S^{\nu}_{a,\mu} & \defeq \delta_{\mu,\AC_d(\nu \cup a)} \sqrt{\frac{1}{q(a|\nu)}} \sqrt{\frac{1}{\sum_{b \in \AC_d(\nu \cup a)} q(b|\nu \cup a)}} , \\
	M               & \defeq \sum_{\nu \pt_d N-2} M^\nu, \quad M^\nu \defeq X^\nu{\tp} X^\nu , \quad X^\nu \defeq H^\nu S^\nu ,                         \\
	q(a | \lambda)  & \defeq \frac{d_{\lambda\cup a}}{N \, d_\lambda},
\end{align}
where the delta function is defined as $\delta_{\mu,\AC_d(\nu \cup a)} \defeq 1$ iff there exists $b \in \AC_d(\nu \cup a)$ such that $\mu = \nu \cup a \cup b$, and $\delta_{\mu,\AC_d(\nu \cup a)} \defeq 0$ otherwise. Note that for every $\lambda$ $q(a | \lambda)$ is a measure over all addable cells $\AC(\lambda)$, i.e.\ $\sum_{a \in \AC(\lambda)} q(a | \lambda) = 1$.

\section{Proof of \texorpdfstring{\cref{thm:two_step_psucc}}{}}
\label{apx:proof_psuc}
\begin{proof}
	Recall that the average probability of success is defined as
	\begin{equation}
		p_{\mathrm{succ}} \defeq \Tr \sof*{\mathcal{N}(I/d \otimes I/d)},
	\end{equation}
	which, using tensor network representation, can be rewritten as
	\begin{equation}
		p_{\mathrm{succ}} = \frac{1}{d^{N+2}} \sum_{i=1}^{N} \sum_{\substack{j=1 \\ j \neq i}}^{N-1} \Tr \sof{\Tr_{\bar{A}_1,A_i} \sof{\sqrt{\Pi_i} O O^\dagger \sqrt{\Pi_i}}
			\Tr_{\bar{A}_2}\sof{ \widetilde{\Pi}_{j}}}
	\end{equation}
	Using equivariance properties, we conclude that
	\begin{equation}
		\label{eq:psuc_tr}
		p_{\mathrm{succ}} = \frac{N}{d^{N+2}}  \Tr \sof[\bigg]{\Tr_{\bar{A}_1,A_N} \sof{\sqrt{\Pi_N} O O^\dagger \sqrt{\Pi_N}}
			\Tr_{\bar{A}_2}\sof[\bigg]{\sum_{j=1}^{N-1} \widetilde{\Pi}_{j}}}
	\end{equation}
	Recall, that our two-step pPBT protocol has the following conditions on the effect operators in the first measurement round for every irrep $\Lambda \in \hat \A_{N,1}^d$:
	\begin{align}\label{def:G_lambda_ppbt}
		\Pi_{0}^\Lambda =
		\begin{cases}
			I &
			\text{if $\Lambda = (\mu,\square)$}, \\
			0 &
			\text{if $\Lambda = (\lambda,\0)$},
		\end{cases}
		\quad\quad\quad
		\sum_{j=1}^N \Pi_{j}^{\Lambda} =
		\begin{cases}
			0 &
			\text{if $\Lambda = (\mu,\square)$}, \\
			I &
			\text{if $\Lambda = (\lambda,\0)$},
		\end{cases}
	\end{align}
	and for the second round for every irrep $\tilde \Lambda \in \hat \A_{N-1,1}^d$ we have
	\begin{align}\label{def:G_lambda_ppbt2}
		\widetilde{\Pi}_{0}^{\tilde \Lambda} =
		\begin{cases}
			I &
			\text{if $\tilde \Lambda = (\lambda,\square)$}, \\
			0 &
			\text{if $\tilde \Lambda = (\nu,\0)$},
		\end{cases}
		\quad\quad\quad
		\sum_{j=1}^{N-1} \widetilde{\Pi}_{j}^{\tilde \Lambda} =
		\begin{cases}
			0 &
			\text{if $\tilde \Lambda = (\lambda,\square)$}, \\
			I &
			\text{if $\tilde \Lambda = (\nu,\0)$}.
		\end{cases}
	\end{align}
	So we can write the second term inside the trace of \cref{eq:psuc_tr} via \cref{lem:partial_trace_E} in the Gelfand--Tsetlin basis for every irrep $\lambda \pt_d N-1$:
	\begin{equation}
		\of[\bigg]{\Tr_{\bar{A}_2}\sof[\bigg]{\sum_{j=1}^{N-1} \widetilde{\Pi}_{j}}}^{(\lambda)} = \sum_{S \in \Paths_{N-1}(\lambda)} \of[\bigg]{\sum_{r \in \RC(\lambda)}\frac{m_{\lambda \setminus r}}{m_\lambda}} \proj{S}
	\end{equation}
	Consider now the first term
	\begin{equation*}
		\frac{1}{d^N} \Tr_{\bar{A}_1,A_N} \sof{\sqrt{\Pi_N} O O^\dagger \sqrt{\Pi_N}},
	\end{equation*}
	and recall how $\Pi_N$ and $O$ are written in the Gelfand--Tsetlin basis:
	\begin{align}
		( \sqrt{\Pi_N})^{(\lambda,\0)}      & = \sum_{S \in \Paths_{N-1} (\lambda)} \frac{\proj{w^\lambda_{S}}}{|| \ket{w^\lambda_{S} }||},                                                             \\ \ket{w^\lambda_{S}} &= \sum_{a \in \AC_d(\lambda)} \sqrt{\frac{d_{\lambda \cup a}}{N d_{\lambda}}} \ket{S \to \lambda \cup a \to (\lambda, \varnothing)}, \\
		|| \ket{w^\lambda_{S} }||^2         & = \sum_{a \in \AC_d(\lambda)} \frac{d_{\lambda \cup a}}{N d_{\lambda}},                                                                                   \\
		\frac{O^{(\lambda,\0)}}{\sqrt{d^N}} & = \sum_{a \in \AC_d(\lambda)} \sqrt{\frac{f_{\lambda \cup a}}{d_{\lambda \cup a} m_{\lambda \cup a}}} \sum_{T \in \Paths_{N} (\lambda \cup a) } \proj{T}.
	\end{align}
	Note that
	\begin{equation}
		\frac{1}{d^N}|| O^{(\lambda,\0)} \ket{w^\lambda_{S} }||^2 = \sum_{a \in \AC_d(\lambda)} \frac{d_{\lambda \cup a}}{N d_{\lambda}} \frac{f_{\lambda \cup a}}{d_{\lambda \cup a} m_{\lambda \cup a}} = \frac{1}{N d_{\lambda}} \sum_{a \in \AC_d(\lambda)} \frac{f_{\lambda \cup a}}{m_{\lambda \cup a}},
	\end{equation}
	so we can easily calculate
	\begin{align}
		\frac{1}{d^N}(\sqrt{\Pi_N} O O^\dagger \sqrt{\Pi_N})^{(\lambda,\0)}
		 & = \frac{1}{N d_{\lambda}} \of[\bigg]{\sum_{a \in \AC_d(\lambda)} \frac{f_{\lambda \cup a}}{m_{\lambda \cup a}}} \sum_{S \in \Paths_{N-1} (\lambda)} \frac{\proj{w^\lambda_{S}}}{|| \ket{w^\lambda_{S} }||^2},
	\end{align}
	where we slightly abused the notation since $|| \ket{w^\lambda_{S} }||^2$ and $|| \ket{O w^\lambda_{S} }||^2$ do not depend on $S$. Then using \cref{lem:partial_trace_E} it is easy to compute:
	\begin{align}
		\of*{\frac{1}{d^N}\Tr_{\bar{A}_1,A_N} \sof{\sqrt{\Pi_N} O O^\dagger \sqrt{\Pi_N}}}^{(\lambda)} & = \frac{1}{N d_{\lambda}} \of[\bigg]{\sum_{a \in \AC_d(\lambda)} \frac{f_{\lambda \cup a}}{m_{\lambda \cup a}}}
		\sum_{S \in \Paths_{N-1} (\lambda)} \of[\bigg]{\sum_{a \in \AC_d(\lambda)} \frac{ d_{\lambda \cup a}}{N d_{\lambda}}} \frac{\proj{S}}{|| \ket{w^\lambda_{S} }||^2}                                               \\
		                                                                                               & =
		\frac{1}{N d_{\lambda}} \of[\bigg]{\sum_{a \in \AC_d(\lambda)} \frac{f_{\lambda \cup a}}{m_{\lambda \cup a}}}
		\sum_{S \in \Paths_{N-1}(\lambda)} \proj{S}.
	\end{align}
	Now, we can combine everything to get the claimed formula for the probability of success:
	\begin{align}
		p_{\mathrm{succ}} & = \frac{N}{d^2} \sum_{\lambda \pt_d N-1}
		\frac{1}{N d_{\lambda}} \of[\bigg]{\sum_{a \in \AC_d(\lambda)} \frac{f_{\lambda \cup a}}{m_{\lambda \cup a}}}  \of[\bigg]{\sum_{r \in \RC(\lambda)}\frac{m_{\lambda \setminus r}}{m_\lambda}} m_\lambda d_\lambda \\
		                  & = \frac{1}{d^2} \sum_{\lambda \pt_d N-1}
		\of[\bigg]{\sum_{a \in \AC_d(\lambda)} \frac{f_{\lambda \cup a}}{m_{\lambda \cup a}}}
		\of[\bigg]{\sum_{r \in \RC(\lambda)} m_{\lambda \setminus r}}.
	\end{align}
	Now we consider the specific choice of all $f_\mu$ coefficients which corresponds to the optimal resource in one round pPBT. Namely, for every $\mu \pt_d N$ we define
	\begin{equation}
		f_\mu = \frac{ m_\mu^2}{\sum_{\mu' \pt_d N} m_{\mu'}^2}.
	\end{equation}
	Using the fact
	\begin{equation}
		d \cdot m_\lambda = \sum_{a \in \AC_d(\lambda)} m_{\lambda \cup a,}
	\end{equation}
	we get
	\begin{align}
		p_{\mathrm{succ}} & = \frac{1}{d} \frac{1}{\sum_{\mu \pt_d N} m_{\mu}^2} \sum_{\lambda \pt_d N-1}  m_\lambda \of[\bigg]{\sum_{r \in \RC(\lambda)} m_{\lambda \setminus r}} \\
		                  & = \frac{1}{d} \frac{1}{\sum_{\mu \pt_d N} m_{\mu}^2} \of[\bigg]{ \sum_{\nu \pt_d N-2} \sum_{a \in \AC_d(\nu)} m_{\nu \cup a} m_{\nu}}                  \\
		                  & = \frac{1}{d} \frac{1}{\sum_{\mu \pt_d N} m_{\mu}^2} \of[\bigg]{d \sum_{\nu \pt_d N-2} m_{\nu}^2}                                                      \\
		                  & = \frac{\sum_{\nu \pt_d N-2} m_{\nu}^2}{\sum_{\mu \pt_d N} m_{\mu}^2} = \frac{\binom{N-2+d^2-1}{d^2-1}}{\binom{N+d^2-1}{d^2-1}}                        \\
		                  & = \frac{N (N-1)}{(N+d^2-1)(N+d^2-2)}.
	\end{align}

\end{proof}

\printbibliography

\end{document}